\documentclass[10pt,a4paper]{article}
\usepackage[utf8]{inputenc}
\usepackage[english]{babel}
\usepackage{amsmath}
\usepackage{amsthm}
\usepackage{float}
\usepackage{amsfonts}
\usepackage{amssymb}
\usepackage{graphicx}
\usepackage{enumerate}
\usepackage[margin=1in]{geometry}
\usepackage{appendix}
\usepackage{authblk}
\usepackage{subcaption}
\usepackage[normalem]{ulem}

\usepackage[textsize=tiny]{todonotes}
\usepackage{hyperref}
\usepackage{xcolor}
\hypersetup{
  colorlinks=true,
  linkcolor={blue!80!black},
  citecolor=black,
  urlcolor={blue!80!black}
}

\theoremstyle{plain}
\newtheorem{theorem}{Theorem}[section]
\newtheorem{lemma}[theorem]{Lemma}
\newtheorem{corollary}[theorem]{Corollary}
\newtheorem{proposition}[theorem]{Proposition}
\newtheorem{assumption}[theorem]{Assumption}

\theoremstyle{definition}
\newtheorem{definition}[theorem]{Definition}

\theoremstyle{remark}
\newtheorem{remark}[theorem]{Remark}
\newtheorem{example}[theorem]{Example}

\numberwithin{equation}{section}

\newcommand{\R}{\mathbb{R}}
\newcommand{\E}{\mathbb{E}}
\newcommand{\G}{{\mathcal G}}
\newcommand{\tr}{{\operatorname{Tr}}}
\newcommand{\cov}{{\operatorname{Cov}}}
\newcommand{\tmu}{\tilde{\mu}}
\newcommand{\tpi}{\tilde{\pi}}
\newcommand{\norm}[1]{\left\|#1\right\|}

\title{Bayesian optimal experimental design with Wasserstein information criteria}
\author[,1]{Tapio Helin\thanks{\href{mailto:tapio.helin@lut.fi}{tapio.helin@lut.fi}}}
\author[,2]{Youssef Marzouk\thanks{\href{mailto:ymarz@mit.edu}{ymarz@mit.edu}}}
\author[,3]{Jose Rodrigo Rojo-Garcia\thanks{\href{mailto:joserodrigo.garcia@wias-berlin.de}{joserodrigo.garcia@wias-berlin.de}}}
\affil[1]{LUT University, Lappeenranta, Finland}
\affil[2]{Massachusetts Institute of Technology, Cambridge, USA}
\affil[3]{Weierstrass Institute for Applied Analysis and Stochastics, Berlin, Germany}

\begin{document}

\maketitle

\begin{abstract}

Bayesian optimal experimental design (OED) provides a principled framework for selecting observations or experiments. 
We introduce new Bayesian design criteria based on the expected Wasserstein-$p$ distance between the prior and posterior distributions, termed Wasserstein information criteria. These criteria have many parallels with the widely used expected information gain (EIG) criterion, which instead relies on the Kullback--Leibler divergence. We show that the Wasserstein-$2$ criterion admits a closed-form solution in the linear-Gaussian setting, a property which can be used for more general approximation schemes, and contrast this solution with classical notions of Bayesian alphabetic optimality. Then we develop a stability analysis of the Wasserstein-$1$ criterion, wherein we bound errors induced by perturbations of the prior or likelihood. We partially extend this analysis to the Wasserstein-$2$ criterion. In particular, these results yield error rates for empirical approximations of the prior. We then illustrate the computability of the Wasserstein-$2$ criterion and demonstrate our approximation rates through simulations.
\end{abstract}


\section{Introduction}

The collection of high-quality data, whether in field studies or laboratory settings, is often constrained by factors such as cost, time, and resource availability. Designing experiments that are both highly informative and efficient is therefore a critical challenge in modern scientific and engineering research. \emph{Optimal experimental design} (OED) offers a systematic framework for this task, by formulating it as the maximization of an (expected) utility function that guides the choice of experiments. Traditionally, OED in large-scale mathematical models has required prohibitive computational effort. However, advances in both algorithms and computational hardware are steadily making design even in large-scale models more tractable, thus motivating interest in OED criteria that encode new experimental goals and in computational approximations that allow robust convergence guarantees.

In this study, we adopt a Bayesian approach to OED \cite{chaloner1995bayesian}, which seeks to maximize a \textit{utility function} that is averaged across the joint distribution of the data and the unknown model parameters. More precisely, suppose $X$ denotes our unknown model parameter, $Y$ denotes the observations, and $\theta$ is the design parameter. The expected utility $U$ is given by
\begin{equation}
    \label{eq:main_utility}
	U(\theta) = \E^{\nu_\theta} u(X,Y; \theta),
\end{equation}
where $u(x, y; \theta)$ denotes the utility of an experiment at conditions $\theta$ that yields observations $y$, when the true parameter value is $x$;
and $\nu_\theta$ is the joint prior distribution of the random variables $X$ and $Y$. Notice that $\nu_\theta$ depends on $\theta$. 
For an extensive recent review, see \cite{huan2024optimal}. 

Our approach here is guided especially by experimental design questions arising in modern inverse problems \cite{engl1996regularization}, which involve the imaging of high-dimensional objects through indirect observations. Over the past two decades, the Bayesian approach to inverse problems has garnered wide attention \cite{kaipio2006statistical, Stu10, DS17}. In Bayesian inverse problems, the likelihood function emerges from an observational model such as
\begin{equation}
	\label{eq:BIP}
	Y = \G(x; \theta) + \mathcal{E}
\end{equation}
where $Y$ and $x$ are observations and parameters, respectively, as before; $\mathcal{E}$ is the measurement noise; and the mapping $\G$ is induced by the complex mathematical model underpinning the inverse problem. For instance, $\mathcal{G}$ might represent the map from the coefficients to the solution of partial differential equation, with the latter observed pointwise at a finite number of locations.

For design problems arising from models such as \eqref{eq:BIP}, perhaps the most studied choice for the utility $u$ in \eqref{eq:main_utility} is the Kullback--Leibler (KL) divergence from the prior to the posterior distribution, $u(x,y; \theta) = u(y; \theta) =\mathcal{D}_{\text{KL}}( \mu_\theta^Y \| \mu)$, where $\mu$ denotes the prior distribution of the unknown $x$, and $\mu_\theta^Y$ the posterior distribution of $x$ conditioned on the observation $Y=y$.
The resulting expected utility $U$ is termed the expected information gain (EIG). The EIG is considered a `fully Bayesian' design objective \cite{ryan16}, as the underlying $u$ depends on the entire posterior density. This choice also has extensive justifications from an information theoretic perspective \cite{lindley1956measure,bernardo1979expected}.
%
Another popular utility is given by the negative squared distance (NSD) between $x$ and the posterior mean for a given observation $y$, $u(x, y ; \theta) = -\|x - \mathbb{E}^{\mu_\theta^Y}X \|^2$. Consequently, the expected utility corresponds to the evidence-averaged trace of the posterior covariance. This approach is not fully Bayesian in the same sense as EIG, as it only focuses on the minimization of posterior variance. (Note that in nonlinear settings, the EIG is \textit{not} only a functional of the posterior covariance.) Both of these utilities are well-studied and extend rigorously to an infinite-dimensional setting \cite{alexanderian2021optimal}, where the unknown $X$ can take values in a function space such as a Hilbert space.

Our work is motivated by the tension between the principled, information-theoretic foundation of EIG and the intuitive, estimator distance-based formulation of NSD. To highlight this tension through practical examples, consider the special case of a linear $\G$ in \eqref{eq:BIP} with Gaussian $\mathcal{E}$ and a Gaussian prior on $X$, leading to Gaussian posterior distribution. It is well-known that EIG and NSD then reduce to the classical D- and A-optimality criteria on the posterior covariance, respectively. More precisely, maximizing EIG corresponds to minimizing the log-determinant of the posterior covariance matrix, whereas maximizing expected NSD corresponds to minimizing the trace of the posterior covariance. As a result, NSD distributes the variance minimization more evenly across all posterior marginals, which can be advantageous in imaging applications. In contrast, EIG presents a more intricate challenge: the expected utility diverges to negative infinity if any one-dimensional marginal posterior distribution becomes overly concentrated (i.e., if a diagonal element of the posterior covariance matrix approaches zero). In principle, this can result in undesirable designs, where, in the extreme case, perfect reconstruction of a single pixel occurs at the cost of significant uncertainty elsewhere.

In this paper, we propose an alternative criterion for Bayesian optimal experimental design: maximizing the averaged Wasserstein-$p$ distance between the posterior and the prior, i.e., 
\begin{equation}
    U_p(\theta) = \E^{\nu_\theta} \, W_p^p(\mu, \mu^Y_\theta), \label{eq:Updef}
\end{equation} 
where $\mu$ and $\mu^Y_\theta$ are the prior and posterior probability distributions of $X$, respectively. 
The Wasserstein distance provides a geometric measure of discrepancy between probability distributions by quantifying the optimal transport cost required to transform one distribution into another. 
Consequently, the expected utility rewards \textit{increased} transport cost associated with updating the prior to the posterior, averaged over the prior marginal of $Y$.
Moreover, in contrast with the KL divergence, the Wasserstein distance remain well-defined even for distributions with non-overlapping support, making it particularly useful for OED with a variety of surrogate or empirical modeling approaches.

\subsection{Our contribution}

Broadly, we will show that the expected Wasserstein utility \eqref{eq:Updef}, termed the `W-optimality' criterion, preserves many desirable features of EIG and NSD, but also combines the best of both worlds.
More specifically, we establish the following results:
\begin{itemize}
    \item The expected Wasserstein utility $U_p$ is a proper Bayesian design criterion, aligned with \cite{ryan16}, taking the full posterior distribution as an input. It also extends to the setting where $X$ takes values in an infinite-dimensional Hilbert space, as proved in Theorem \ref{thm:well-posedness}. In Section \ref{subsec:W_information_criterion}, we discuss the validity of $U_p$ as an information criterion in the sense proposed by Ginebra \cite{ginebra2007}.

    \item In the linear--Gaussian case, the $U_2$ optimality criterion has an explicit formula given in Theorem \ref{thm:linear_case}, involving the prior and posterior covariances. This formula bears similarity to the Bayesian A-optimality criterion but distributes the variance objective across posterior marginals, weighted by the prior. In Section \ref{subsec:W2_Gaussian}, we discuss further connections that arise when considering weighted NSD and weighted Wasserstein distances.
    
    \item We prove that the $W_1$-optimality criterion is stable to likelihood and prior perturbations. In particular, prior stability demonstrates a major advantage of the Wasserstein utility, as it enables quantification of errors when applying \textit{empirical} approximations. To our knowledge, similar results are not available for the EIG or NSD.
 
    \item  Borrowing ideas from the optimal transport literature, we demonstrate computability of the W-optimality criterion and, in Section \ref{sec:algorithms}, propose a numerical scheme to evaluate the $U_2$ criterion. Furthermore, we demonstrate the prior stability results obtained for $U_1$ through one-dimensional examples, where the numerical convergence aligns with the expected theoretical rates when applying empirical prior approximations.
\end{itemize}

We also mention that as a side product of our stability theory, we prove new posterior stability results for Bayesian inversion in Wasserstein-$1$ distance both for likelihood (Theorem \ref{thm:U1_like_stability} claim $(i)$) and prior (Theorem \ref{thm:U1_prior_stability} claim $(i)$) perturbations. Our results improve the state of the art by quantifying the perturbation over the evidence expectation, which is not available in previous results such as \cite{sprungk2020local, garbuno2023bayesian}.

\subsection{Literature overview}

Bayesian optimal experimental design has an extensive literature and a rich history. For comprehensive recent overviews, we refer the reader to \cite{huan2024optimal, rainforth2024modern, ryan16}. A broad discussion on different utilities is given in the classical reference \cite{chaloner1995bayesian}. Moreover, a general notion of valid information measures is proposed in \cite{ginebra2007}.

The stability analysis provided in this paper is aligned with earlier work by the authors in \cite{duong2023stability} in the context of EIG and likelihood perturbations. Stability of EIG has also been studied from the point of view of variational approximations to the underlying distributions, motivated by computation; specifically, Foster and others explore variational approximations \cite{foster2019variational, foster2020unified} to overcome nested integration challenges related to estimation of the EIG criterion. 

Inverse problems constitute a class of high-dimensional inference problems where complex mathematical models such as partial differential equations relate unknown parameters to observable data. The need for scalability across various discretization levels in inverse problems has motivated work extending traditional Bayesian experimental design criteria to infinite-dimensional settings. Indeed, EIG and NSD criteria have well-defined corresponding formulations in the nonparametric setting, reviewed, e.g., in \cite{alexanderian2021optimal}.
The development of numerical BOED methods for inverse problems has gained substantial attention during the last decade or so. Among recent advances, we mention work accelerating standard nested Monte Carlo algorithms for EIG \cite{kaarnioja2024quasimontecarlobayesiandesign, beck2020_mulvl, beck2018fast, goda2022}, the use of learning-based surrogates \cite{koval2024tractable, cui2025subspace, wu2023large, o2024derivative, go2025sequential, dong2025variational, siddiqui2024deep, li2024expectedinformationgainestimation}, and other approximation schemes \cite{hoang2025scalable, alexanderian2024optimal, lartaud2024sequential, helin2022edge, burger2021sequentially}. There is also ongoing interest in non-Bayesian OED formulations, as in  \cite{aarset2024global, chenglie2024, encinar2024optimal}. 

Similar to OED, the literature around Wasserstein distances and their fundamental connection to optimal transport theory is extensive. For a comprehensive treatment of the topic, we refer to Villani's seminal work \cite{villani2008optimal}. Computational aspects of optimal transport, including efficient algorithms for computing Wasserstein metrics, are thoroughly discussed in \cite{peyre2020computational}. Furthermore, Wasserstein distances have emerged as powerful tools across various statistical applications, including hypothesis testing, density estimation, and Bayesian inference, as systematically reviewed in \cite{panaretos2019statistical}.

The paper \cite{kerrigan2025geometric} appeared after the first version of this manuscript was posted on the arXiv and also considers Wasserstein objectives for OED. Specifically, the authors propose using an optimal transport distance between $\nu_\theta$ and the product of its $X$ and $Y$ marginals, $\nu_{X} \otimes \nu_{Y,\theta}$, as a design criterion. A modified criterion proposed in \cite{kerrigan2025geometric}, called the `target transport distance,' is precisely our Wasserstein information criterion in the case of $L^p$ cost. The authors argue that such criteria provide useful ways of encoding scale and other aspects of geometry in a design problem.


This paper is organized as follows. In Section \ref{sec:Wasserstein} we discuss the construction and basic properties of the Wasserstein distance based information criterion. In particular, Section \ref{subsec:W2_Gaussian} considers the Wasserstein-$2$ based utility with Gaussian posterior and prior distributions, where explicit formulas can be established. Section \ref{sec:stability} is devoted to the stability analysis of the expected utility. In Section \ref{sec:algorithms} we demonstrate our theoretical findings for $p=1$ through simulations and propose a numerical scheme for approximating the Wasserstein-$2$ criterion. Finally, Section \ref{sec:conclusions} presents our conclusions and outlines directions for future work.

\section{Wasserstein distance and information criteria}
\label{sec:Wasserstein}

This section introduces the Wasserstein information criterion and examines its mathematical well-posedness. We then analyze the special case of Gaussian distributions, which yields closed-form expressions for the expected utility and, consequently, significantly simplifies computational implementation. 

\subsection{Wasserstein information criterion}
\label{subsec:W_information_criterion}

In what follows, ${\mathcal P}({\mathcal X})$ denotes the Borel probability measures on ${\mathcal X}$ and $M_p(\mu) = \int_{\mathcal X} \norm{x}^p_{\mathcal X} \mu(dx)$ denotes the $p$-th moment of a measure $\mu \in {\mathcal P}({\mathcal X})$. We write $\mathcal{P}_p({\mathcal X}) \subset {\mathcal P}({\mathcal X})$ for the subset of probability measures with finite $p$-th moment.

\begin{definition}
    Given two probability measures $\mu_{1},\mu_{2}\in \mathcal{P}(\mathcal{X})$. For $p\in [1,\infty)$ the Wasserstein distance is defined as
    \begin{equation}
    \label{Wasserstein_distance_def}
    W_p(\mu_{1},\mu_{2})=\left(\inf _{\gamma \in \Gamma(\mu_{1},\mu_{2})} \int_{\mathcal{X} \times \mathcal{X}}\|x-w\|^p \mathrm{~d} \gamma(x, w)\right)^{1 / p},
    \end{equation}
    where $\Gamma(\mu_{1},\mu_{2})$ is the set of all couplings between $\mu_{1}$ and $\mu_{2}$.
    
\end{definition}

Some basic properties of the Wasserstein distance are listed in the next proposition.
\begin{proposition}[\cite{villani2008optimal}]
    \label{prop:general_properties_Wp}
    Let $q \geq p \geq 1$, and $\mu_1, \mu_2, \mu_3 \in {\mathcal P}(X)$. The following statements hold:
    \begin{itemize}
        \item[(1)] $W_p$ is a metric in $\mathcal{P}(\mathcal{X})$.        
        \item[(2)] $W_{p}(\mu_1,\mu_2) \leq W_{q}(\mu_1,\mu_2)$.  
        \item[(3)] $W_{p}^{p}(\mu_1,\mu_2) \leq 2^{p-1} (M_{p}(\mu_1)+M_{p}(\mu_2))$.
    \end{itemize}  
\end{proposition}
We note that property (3) indicates that $W_p$ is a finite distance in ${\mathcal P}_p(X)$.
Next, the following well-known theorem will be the basis for part of our numerical implementations in Section \ref{sec:algorithms}.
\begin{theorem}[{\cite[Brenier's Theorem]{villani2008optimal}}]
    \label{Optimal coupling}
    Let $\mu_1,\mu_2\in\mathcal{P}_{p}(\mathcal{X})$ be probability measures. 
    There exists an optimal coupling $\gamma^{*}\in\Gamma(\mu_1,\mu_2)$, such that
    \begin{equation}         
        W_p^p(\mu_1,\mu_2) = \int_{\mathcal{X} \times \mathcal{X}}\|x_1-x_2\|^p \gamma^{*}(dx_1,dx_2).
    \end{equation}
    Moreover, if $\mathcal{X} = \mathbb{R}^n$, $p = 2$ and $\mu_1$ is absolutely continuous respect to the Lebesgue measure, then there exists a unique (up to a constant) convex and almost everywhere differentiable potential function $\varphi$ such that $T(x) = \nabla \varphi(x)$ and
    \begin{equation}
        \label{eq:optimal_coupling}
        W_2^2(\mu_1, \mu_2) = \int_{\R^n} \norm{x-T(x)}^2 \mu_1(dx).
    \end{equation}
\end{theorem}

We formulate our Bayesian inference framework as follows: Let $x$ denote the unknown parameter of interest, taking values in the separable Hilbert space ${\mathcal X}$. We endow $x$ with a prior distribution $\mu \in {\mathcal P}_p({\mathcal X})$. The parameter $x$ is observed through measurements $Y \in \mathbb{R}^d$, which follow a conditional probability density $\pi(\cdot | x; \theta)$, with $\theta$ representing the experimental design variable. The marginal density of $Y$ is denoted by $\pi(\cdot; \theta)$.

Now we are prepared to define our information criteria utilizing Wasserstein distance. 

\begin{definition}
\label{def:wp_utility}
Let $\mu \in {\mathcal P}_p(X)$. The expected Wasserstein-$p$ utility, $p\in [1,\infty)$, of the design $\theta$ is defined as
\begin{equation}
	\label{eq:wp_utility}
    U_p(\theta) := \E^{\nu_\theta}W_{p}^{p}(\mu,\mu^Y(\cdot; \theta)).
    \end{equation}
\end{definition}
Clearly, as the Wasserstein distance in \eqref{eq:wp_utility} does not depend on variable $x$, the expectation reduces to an expectation over the marginal density $\pi(\cdot; \theta)$. In what follows, for notational simplicity, we suppress the explicit dependence on the design variable $\theta$ in expressions where this dependence remains constant throughout the analysis and the result applies uniformly for all designs.

\begin{theorem}[Well-posedness]
\label{thm:well-posedness}
Let $\mu \in {\mathcal P}_p(X)$. Then the expected Wasserstein-$p$ utility in \eqref{eq:wp_utility} for $p \in [1,\infty)$ is finite.
\end{theorem}
\begin{proof}
The claim follows directly from the upper bound 
	$U_p \leq 2^{p-1}\left(M_p(\mu) + \E^{\pi(y)}M_p(\mu^y)\right) = 2^p M_p(\mu) < \infty$.
\end{proof}

Ginebra \cite{ginebra2007} proposes a general formalism for what comprises as a valid measure of information in a statistical experiment and, consequently, a valid optimality design criteria. In this formalism, an information measure must satisfy a minimal set of requirements: $(i)$ it is real-valued; $(ii)$ it returns zero for a `totally non-informative experiment,' where $Y$ is independent of $X$; and $(iii)$ it satisfies sufficiency ordering \cite{blackwell1951comparison,blackwell1953equivalent,lecam1964sufficiency}.

Lehmann and Casella \cite{lehmann2006theory} phrase statistical sufficiency as follows: We say $Y|x,\theta_1$ is sufficient for (or `always at least as informative as') $Y|x,\theta_2$ (with the same parameter value $x$) if there exists a random variable $\eta$ with a known probability distribution and a function $W$ such that $W(Y,\eta)|x,\theta_1$ has the same distribution as $Y|x,\theta_2$ for all $x$. 

A sufficiency ordering then implies that if an experiment with design $\theta_1$ is sufficient for an experiment with design $\theta_2$, it follows that $U_p(\theta_2) \leq U_p(\theta_1)$. In other words, condition (iii) states that a valid information criterion must preserve this partial ordering of probability measures, ensuring that more informative experimental settings receive higher utility scores. Many classical criteria, including EIG, satisfy this generalized notion of information measure. 

The Wasserstein-$p$ utility clearly satisfies first two conditions, in particular, as under a non-informative observation we have $\mu^y = \mu$. The third condition 
requires connecting sufficiency ordering to convex ordering by the Blackwell--Sherman--Stein theorem \cite{blackwell1951comparison,blackwell1953equivalent,lecam1964sufficiency}. A special implication of this theory is given by Ginebra \cite{ginebra2007}.

\begin{proposition}[{\cite[Prop. 3.2]{ginebra2007}}]
Suppose ${\mathcal X} = (x_1, ..., x_k)$ is finite. Then experiment $\theta_1$ is sufficient for $\theta_2$ if and only if for a given strictly positive prior distribution $\mu$ on ${\mathcal X}$, we have
\begin{equation*}
    \E^{Y|\theta_1} \phi(\mu^y(\cdot; \theta_1))
    \geq \E^{Y|\theta_2} \phi(\mu^y(\cdot; \theta_2))
\end{equation*}
for every convex function $\phi$ on the simplex of $\R^k$, where $\mu^y(\cdot; \theta)$ is also understood as an element of this simplex.
\end{proposition}

Since the mapping $\mu \mapsto W_p(\mu, \nu)^p$ is convex (see e.g., \cite[Thm. 4.8]{villani2008optimal}), we obtain as a direct consequence that the utility function $U_p$ satisfies the sufficiency ordering condition for finite parameter spaces $\mathcal{X}$. Ginebra \cite[p. 178]{ginebra2007} outlines a framework for extending this result to countable and continuous parameter spaces by utilizing Le Cam's theory of statistical experiments \cite{lecam2012asymptotic} and generalizations of convex ordering to stochastic processes in Bassan and Scarsini \cite{bassan1991convex}.
However, a rigorous proof of this extension remains beyond the scope of the current paper.

\subsection{Wasserstein-$2$ utility and Gaussian posterior}
\label{subsec:W2_Gaussian}

The presence of Gaussian distributions significantly simplifies the expression of $U_2$-utility, as it is well-known that the Wasserstein-2 distance between two Gaussian measures can be computed explicitly \cite{gelbrich1990formula}.
Namely, suppose $\mu_1$ and $\mu_2$ are Gaussian measures on a separable Hilbert space ${\mathcal X}$, with $\mu_j$ having mean $m_j \in {\mathcal X}$ and covariance operator $C_j$, for $j = 1, 2$.
Then it holds that
\begin{equation}
\label{eq:W2_for_Gaussians}
W_2^2\left(\mu_1,\mu_2\right) = \left\|m_1-m_2\right\|^{2}+\tr_{\mathcal X}\left(C_1+C_2-2\left(C_1^{1 / 2} C_2 C_1^{1/2}\right)^{1 / 2}\right). 
\end{equation}
Consider now the inverse problem \eqref{eq:BIP} with a linear forward operator $\G : {\mathcal X} \to \R^d$ and a centred Gaussian prior $\mu_0$ with covariance $C_0$.
It is well-known (see e.g. \cite{Stu10}) that the posterior distribution $\mu^y$ is Gaussian with mean $x_{\text{post}}(y)$ and covariance $C_{\text{post}}$ satisfying \begin{equation}
\label{Gaussian_Posterior}
x_{\text{post}}(y) =C_{\text{post}} \G^* \Gamma^{-1} y \quad \text{and}  \quad C_{\text{post}}=C_0 - C_0 \G^* \left(\Gamma + \G C_0 \G^*\right)^{-1} \G C_0. 
\end{equation}
This enables us to prove the following useful identity.

\begin{theorem}
\label{thm:linear_case}
Consider the inverse problem \eqref{eq:BIP} with $\epsilon \sim {\mathcal N}(0,\Gamma)$. Moreover, suppose that $\G$ is linear and the prior $\mu_0$ is a centred Gaussian with covariance $C_0$,
we have that the expected utility $U_2$ in  satisfies
\begin{equation*}
	U_2=2 \tr_{\mathcal X}\left(C_0\right)-2 \tr_{\mathcal X}\left(\left(C_0^{\frac 12} C_{\text{post}} C_0^{\frac 12}\right)^{\frac 12}\right).
\end{equation*}
\end{theorem}
\begin{proof}
Due to identity \eqref{eq:W2_for_Gaussians}, we have
\begin{equation}
	\label{eq:linear_aux1}
U_2=\E^{\pi(y)}\norm{x_{\text{post}}(y)}^2+\tr_{\mathcal X}\left(C_0+C_{\text{post}}-2\left(C_0^{1 / 2} C_{\text{post}} C_0^{1 / 2}\right)^{1 / 2}\right).
\end{equation}
The evidence distribution $\pi(y)$ is a centred Gaussian with $\cov(y) = \Gamma + \G C_0 \G^*$ and, therefore, combining with \eqref{Gaussian_Posterior} we obtain
\begin{equation}
	\label{eq:linear_aux2}
\E^{\pi(y)}\left\|x_{\text{post}}(y)\right\|^2 = \E^{\pi(y)}\left\|C_{\text{post}} \G^* \Gamma^{-1} y \right\|^2
= \tr_{\mathcal X}\left(C_{\text{post}} \G^* \Gamma^{-1} \cov(y)  \Gamma^{-1}\G C_{\text{post}}\right).
\end{equation}
For convenience, let us abbreviate $B=\Gamma^{-\frac 12} \G C_0^{\frac 12}$ and
\begin{equation*}
	T = I - B^* (I+B B^*)^{-1} B,
\end{equation*}
which yields the expressions
\begin{equation*}
	\cov(y) = \Gamma^{\frac 12} (I+BB^*) \Gamma^{\frac 12} \quad \text{and} \quad
	C_{\text{post}} = C_0^{\frac 12} T C_0^{\frac 12}.
\end{equation*}
This yields us
\begin{align}
	\label{eq:linear_aux3}
	C_{\text{post}} & \G^* \Gamma^{-1} \cov(y)  \Gamma^{-1}\G C_{\text{post}} \nonumber\\
	& = C_0^{\frac 12} TB^*(I+BB^*) BT C_0^{\frac 12} \nonumber\\
	& = C_0^{\frac 12} (I - B^* (I+B B^*)^{-1} B)B^*(I+BB^*) B(I - B^* (I+B B^*)^{-1} B) C_0^{\frac 12} \nonumber\\
	& = C_0^{\frac 12} B^* (I- (I+B B^*)^{-1}BB^*) (I+BB^*) (I-BB^*(I+B B^*)^{-1}) BC_0^{\frac 12} \nonumber\\
	& = 	C_0^{\frac 12} B^*(I+BB^*)^{-1}(I+BB^*)(I+BB^*)^{-1} B C_0^{\frac 12}  \nonumber\\
	& = C_0^{\frac 12} B^*(I+BB^*)^{-1}B C_0^{\frac 12} \nonumber\\
	& = C_0 - C_{\text{post}}.
\end{align}
In consequence, combining identities \eqref{eq:linear_aux1}, \eqref{eq:linear_aux2} and \eqref{eq:linear_aux3}, we obtain the result.
\end{proof}

Theorem \ref{thm:linear_case} establishes that optimizing the $U_2$-utility criterion is mathematically equivalent to minimizing the geometric mean between posterior and prior covariance operators. \\

{\bf Weighted A-optimality.} When the primary objective of an experiment is to obtain a point estimate of the parameters, a canonical utility for the design task is given by the averaged NSD with respect to the specific estimator \cite{chaloner1995bayesian}. When the posterior mean is utilized, the task is then to maximize
\begin{equation}
    U_A = - \E^\nu \norm{A(x- x_{\text{post}}(y))}^2,
\end{equation}
where the linear bounded operator $A : {\mathcal X} \to {\mathcal X}$ serves as a weighting function and $x_{\text{post}}$ is the posterior mean. It is well-known that for the posterior emerging in \eqref{Gaussian_Posterior} it holds that 
$U_A = - \tr_{\mathcal X} (A C_{\text{post}} A^*)$.
Therefore, for the weight induced by the prior $A = C_0^{\frac 12}$, we observe that the utility satisfies
\begin{equation}
    U_{C_0^{1/2}} = - \tr_{\mathcal X} (C_0^{\frac 12} C_{\text{post}} C_0^{\frac 12}).
\end{equation}
This yields a clear connection to the Wasserstein utility $W_2$, where the maximization task is reduced to 
\begin{equation*}
    U_2 = \text{const} - 2 \tr_{\mathcal X} \left(C_0^{\frac 12} C_{\text{post}} C_0^{\frac 12}\right)^{\frac 12}.
\end{equation*}
In the finite-dimensional setting ${\mathcal X}=\R^n$, the two utilities are determined by the generalized eigenvalue problem
\begin{equation}
    \label{eq:gen_eigen_prob}
    C_{\text{post}} w_j = \lambda_j C_0^{-1} w_j
\end{equation}
for $j\geq 1$. Indeed, for any eigenpair $(\lambda_j, z_j)$ of the matrix $C_0^{\frac 12} C_{\text{post}} C_0^{\frac 12}$ we have that $(\lambda_j, w_j)$ with $w_j = C_0^{-\frac 12}z_j$ satisfies \eqref{eq:gen_eigen_prob}. Therefore, we obtain
\begin{equation}
    \label{eq:weightedAopt}
    U_2 = \text{const} - 2 \sum \sqrt{\lambda_j} \quad \text{and} \quad
    U_{C_0^{1/2}} = -\sum \lambda_j.
\end{equation}
Equation \eqref{eq:weightedAopt} reveals a compelling parallel: while $U_{C_0^{1/2}}$ minimizes the sum of marginal variances, $U_2$ minimizes the sum of marginal standard deviations of the prior-weighted posterior covariance $C_0^{\frac{1}{2}} C_{\text{post}} C_0^{\frac{1}{2}}$.\\

{\bf Weighted Wasserstein-$2$ distance.} It is also interesting to consider implications of a weighted transport cost in \eqref{Wasserstein_distance_def}. To that end, let us define
\begin{equation*}
        W_{2,B}(\mu_{1},\mu_{2})=\left(\inf _{\gamma \in \Gamma(\mu_{1},\mu_{2})} \int_{\mathcal{X} \times \mathcal{X}}\|B(x-w)\|^2 \mathrm{~d} \gamma(x, w)\right)^{1 / 2}
\end{equation*}
for a finite-dimensional domain ${\mathcal X} = \R^n$ and an invertible matrix $B\in \R^{n\times n}$. For Gaussian measures $\mu_1$ and $\mu_2$, we can deduce, analogously to the standard case (see e.g., \cite{givens1984class}), that the infimum in the weighted Wasserstein-$2$ distance is attained by a Gaussian coupling $\gamma_* \in \Gamma(\mu_1,\mu_2)$.
Therefore, we have
\begin{equation*}
    W_{2,B}^2(\mu_1, \mu_2) = \E^{\gamma_*}\norm{B(X-W)}^2 = \E^{\tilde{\gamma}_*}\|\widetilde{X}- \widetilde{W}\|^2,
\end{equation*}
where the second equality follows from a change of variables, with coupling $\tilde{\gamma}_*$ having marginals $\widetilde{X} \sim B_{\sharp}\mu_1$ and $\widetilde{W} \sim B_{\sharp}\mu_2$. 

Following the steps in Theorem \ref{thm:linear_case}, we observe
\begin{equation*}
    U_{2,B} := \E^{\pi(y)} W_{2,B}^2(\mu,\mu^y) 
    = 2 \tr\left(BC_0B^\top\right)-2 \tr\left(\left((BC_0B^\top)^{\frac 12} BC_{\text{post}}B^\top (BC_0B^\top)^{\frac 12}\right)^{\frac 12}\right)
\end{equation*}
and, consequently, for the choice $B=C_0^{-\frac 12}$, we have
\begin{equation*}
    U_{2,C_0^{-1/2}} = \text{const} - 2 \tr \left(C_0^{-\frac 12} C_{\text{post}} C_0^{-\frac 12}\right)^{\frac 12}
\end{equation*}
Similar to the previous example, the maximization of $U_{2,C_0^{-1/2}}$ holds similarity with A-optimality criterion $U_{C_0^{-1/2}} = - \tr (C_0^{-1/2} C_{\text{post}} C_0^{-1/2})$ and the difference can be identified as the minimization of the sum of marginal variances versus the sum of marginal standard deviations of the \emph{prior precision-weighted} posterior covariance.

\section{Stability bounds with Gaussian likelihood}
\label{sec:stability}

In this section, we consider the stability of the expected utility \eqref{eq:wp_utility} with respect to perturbation of the likelihood and prior distributions. We focus on the specific case with Gaussian likelihood that emerges from the observational model \eqref{eq:BIP} with Gaussian noise distribution $\epsilon \sim {\mathcal N}(0,\Gamma)$. Notice that for convenience we require $\Gamma$ to be invertible, i.e., the distribution of $\epsilon$ cannot be degenerate in any subspace of $\R^d$. This setup gives rise to the likelihood energy functional
\begin{equation}
	\label{eq:Phigauss}
	\Phi(x, y) = \frac 12 \norm{\G(x) - y}_\Gamma^2,
\end{equation}
where the weighted norm is defined as $\norm{\cdot}_\Gamma = \norm{\Gamma^{-\frac 12}\cdot}$. 

Before proceeding, let us state the assumptions that will be required throughout this section.

\begin{assumption}
\label{ass:prior_perturbation}
The following conditions hold for the mapping $\G : {\mathcal X} \to \R^d$ and the Borel probability measure $\mu$ on ${\mathcal X}$:
\begin{itemize}
	\item[(i)](Lipschitz) There exists $L_1>0$ such that
	\begin{equation*}
		\norm{\G(x) - \G(x')}_\Gamma \leq L_1 \norm{x-x'}
	\end{equation*}
	for all $x,x' \in {\mathcal X}$.
	\item[(ii)](sub-Gaussian prior) There exists $L_2>0$ such that
	\begin{equation*}
		\E^\mu \exp\left(L_2 \norm{x}^2\right) < \infty.
	\end{equation*}
	\item[(iii)]($\G$ is proper) There exists $R,L_3>0$ such that
	$\mu(B(0,R))>0$ and $\sup_{x\in B(0,R)} \norm{\G(x)}_\Gamma < L_3.$
\end{itemize}
\end{assumption}

A central assumption in the following statements is that the constants $L_1$ and $L_2$ satisfy a bound of the form $L_1^2 < C L_2$, where $C$ is a universal constant. To provide context for this assumption, we briefly present two illustrative examples.

\begin{example}
Let $\mu$ be a probability measure on ${\mathcal X}$ such that $\mu(B(0,R_0))=1$ for some $0<R_0<\infty$. Then condition $(ii)$ in Assumption \ref{ass:prior_perturbation} holds for any $L_2>0$. In such a case, for results that follow, the requirement of Lipschitz continuity of $\G$ could be restricted to $B(0,R_0)$. For instance, the uniform measure \cite{Stu10} satisfies such a condition.
\end{example}

\begin{example}
Suppose $\epsilon$ has zero-mean Gaussian statistics with $\Gamma = \delta^2 \Gamma_0$, where $\delta>0$ represents the noise level. Then the condition $(i)$ in Assumption \ref{ass:prior_perturbation} is equivalent to
\begin{equation*}
    \norm{\G(x) - \G(x')}_{\Gamma_0} \leq \delta L_1 \norm{x-x'}
\end{equation*}
for all $x,x'\in {\mathcal X}$. Now consider the effect of decreasing the noise level. For a given mapping $\G$, condition (i) clearly implies stronger contractive as $\delta$ decreases, with the smallest admissible $L_1$ therefore scaling proportionally to $1/\delta$. As a result, the condition $L_1^2 < C L_2$ will be violated for $\delta$ small enough.
\end{example}

{\bf Likelihood perturbations.} Consider two forward mappings $\G$ and $\G_*$ giving rise to corresponding likelihood distributions through observational model \eqref{eq:BIP}. Moreover, let us denote the corresponding posterior distributions by $\mu^y$ and $\mu^y_*$ and utilities by $U_1$ and $U^*_1$, respectively. We have the following result:

\begin{theorem}
\label{thm:U1_like_stability}
Suppose $\G$ and $\G_*$ satisfy Assumption \ref{ass:prior_perturbation} with probability measure $\mu$ on ${\mathcal X}$ and with same constants $L_1, L_2, L_3$ and $R$. Moreover, we assume that $L_1^2< \frac{\sqrt 2 - 1}{2} L_2$. Then there exist constants $K_1$ and $K_2$ depending on $L_1, L_2, L_3, R, \Gamma$ and $d$ such that the following claims hold:
\begin{itemize}
    \item[(i)] The evidence averaged posterior perturbation is bounded by
\begin{equation}
    \E^{\pi(y)} W_1(\mu^y, \mu_*^y) \leq K_1 \left(\E^\mu \norm{\G(x) - \G_*(x)}_\Gamma^2\right)^{\frac 12}.
\end{equation}
\item[(ii)] The perturbation of the expected $U_1$-utility satisfies
\begin{equation*}
	|U_1 - U_1^*| \leq K_2 \left(\E^\mu \norm{\G(x) - \G_*(x)}_\Gamma^2\right)^{\frac 12}.
\end{equation*}
\end{itemize}
\end{theorem}

The proof of Theorem \ref{thm:U1_like_stability} is covered in Section \ref{subsec:like_perturbation}. Notice that the result coincides (up to a constant) with similar likelihood stability result obtained for the EIG in \cite[Thm. 4.4]{duong2023stability}.\\

{\bf Prior perturbations.} Consider now two different prior distributions $\mu$ and $\tmu$ giving rise to corresponding posterior distributions when merged with the likelihood obtained through observational model \eqref{eq:BIP}. Moreover, let us denote the corresponding posterior distributions by $\mu^y$ and $\tmu^y$ and utilities by $U_1$ and $\widetilde{U}_1$, respectively. We have the following result:

\begin{theorem}
\label{thm:U1_prior_stability}
Suppose ${\mathcal G}$ satisfies Assumption \ref{ass:prior_perturbation} with two probability measures $\mu$ and $\tmu$ on ${\mathcal X}$, and $L_1^2< \frac{\sqrt 3 -1}2 L_2$. Then there exist constants $K_1$ and $K_2$ depending on $L_1, L_2, L_3, R, \Gamma$ and $d$ such that the following claims hold:
\begin{itemize}
    \item[(i)] The evidence averaged posterior perturbation is bounded by
    \begin{equation*}
	\E^{\pi(y)} W_1(\mu^y, \tmu^y) \leq K_1 W_2(\mu, \tmu),
\end{equation*}
    \item[(ii)] The perturbation of the expected $U_1$-utility satisfies
    \begin{equation*}
	|U_1 - \widetilde U_1| \leq W_1(\mu,\tmu) + K_2 W_2(\mu, \tmu).
\end{equation*}
\end{itemize}
\end{theorem}

The proof of Theorem \ref{thm:U1_prior_stability} is postponed to Section \ref{subsec:prior_perturbation}.

\begin{remark}
The 'loss' in the Wasserstein distance bounds from $W_1$ to $W_2$ in Theorem \ref{thm:U1_prior_stability} arises from the application of Cauchy-Schwarz inequalities in the respective proofs. Notably one can replace the Cauchy-Schwarz argument with a H\"older type inequality and improve the bounding distance to $W_p$ for $1<p<2$. However, this comes at the cost of requiring stronger moment bounds in Assumption \ref{ass:prior_perturbation} with potential blow-up as $p$ approaches $1$. We leave further exploration of this improvement to future studies.
\end{remark}

Similar proof technique utilized for Theorem \ref{thm:U1_prior_stability} can be also tested with the $U_2$-utility. However, posterior stability in $(i)$ relies strongly on the Kantorovich duality formulation of the $W_1$-distance. Moreover, to the authors best knowledge, there are no stability results for $W_2^2(\mu^y, \tmu^y)$ available in the literature even for fixed observational data $y$. While it is beyond the scope of this paper to provide such a general estimate, we can state the following upper bound.

\begin{theorem}
\label{thm:stability:main}
Suppose ${\mathcal G}$ satisfies Assumption \ref{ass:prior_perturbation} with two probability measures $\mu$ and $\tmu$ on ${\mathcal X}$, and $L_1^2< \frac{\sqrt 3 -1}2 L_2$. Then there exist constants $K_1$ and $K_2$ depending on depending on $L_1, L_2, L_3, R, \Gamma$ and $d$ such that the following inequality holds    
\begin{equation}
	|U_2 - \widetilde U_2| \leq  K_1 W_2(\mu,\tmu) + K_2 \sqrt{\E^{\pi(y)} W_2^2(\mu^y, \tmu^y)}.
\end{equation}
\end{theorem}

In what follows, we break down the proofs for Theorems \ref{thm:U1_like_stability}, \ref{thm:U1_prior_stability} and \eqref{thm:stability:main} into three parts. In Section \ref{subsec:basic_properties} we record some basic inequalities related to the likelihood induced by \eqref{eq:Phigauss} in concert with mappings $\G$ and probability measure $\mu$ satisfying Assumption \ref{ass:prior_perturbation}. Section \ref{subsec:like_perturbation} presents the proof for the likelihood perturbation case, while Section \ref{subsec:prior_perturbation} outlines how the result is obtained for the prior perturbation case.

\subsection{Basic properties}
\label{subsec:basic_properties}

Below, we occasionally apply the Assumption \ref{ass:prior_perturbation} in the form of following Corollary:
\begin{corollary}
\label{cor:direct_cor_assump}
Let ${\mathcal G}$ satisfy Assumption \ref{ass:prior_perturbation} for a probability measure $\mu$ on ${\mathcal X}$.
It holds that
\begin{itemize}
    \item[(i)] for any $x\in {\mathcal X}$ we have
    \begin{equation}
        \label{eq:G_bound}
        \norm{{\mathcal G}(x)}_\Gamma \leq L_1 \norm{x} + L_1R + L_3 \quad \text{and}
    \end{equation}
    \item[(ii)] for any $p\geq 0$ and $L_2'<L_2$ we have
    \begin{equation}
        \E^\mu \left(\norm{x}^p \exp \left(L_2' \norm{x}^2\right)\right) < \infty.
    \end{equation}
\end{itemize}
\end{corollary}

\begin{proof}
For the first claim, we have by Assumption \ref{ass:prior_perturbation} $(i)$ that
\begin{equation*}
	\norm{\G(x)}_\Gamma \leq \norm{\G(x) - \G(x_0)}_\Gamma + \norm{\G(x_0)}_\Gamma
    \leq L_1 \norm{x}+L_1\norm{x_0} + \norm{\G(x_0)}_\Gamma
\end{equation*}
for any $x_0\in{\mathcal X}$. Since $x_0$ is arbitrary, we have
\begin{equation*}
    \norm{\G(x)}_\Gamma \leq L_1 \norm{x} + \inf_{x'\in{\mathcal X}} \left(L_1\norm{x'} + \norm{\G(x')}_\Gamma\right)
\end{equation*}
Due to condition (iii), the infimum on the right-hand side can be bounded by $L_1R+L_3$. 

The second claim follows directly by the H\"older inequality.
\end{proof}

Let us next record some basic properties related to the negative log-likelihood $\Phi$ in \eqref{eq:Phigauss}.

\begin{lemma}
\label{lem:basic_quad_zzy}
Let $y, z_1, z_2 \in \R^d$. We have for any $\tau>0$ that
\begin{multline}
    \label{eq:like_basic_ineq}
    \left|\exp\left(-\frac 12 \norm{z_1 - y}_\Gamma^2\right) -\exp\left(-\frac 12 \norm{z_2 - y}_\Gamma^2\right)\right| \\
    \leq \frac 12 \left(\norm{z_1}_\Gamma + \norm{z_2}_\Gamma + 2\norm{y}_\Gamma \right)  \exp\left(\frac{\tau-1}{2} \norm{y}_\Gamma^2 + \frac 1{2\tau}\norm{z_1}_\Gamma^2 + \frac 1{2\tau}\norm{z_2}_\Gamma^2\right)  \norm{z_1-z_2}_\Gamma
\end{multline}
\end{lemma}

\begin{proof}
We first observe that for $a,b\geq 0$ it follows that 
\begin{equation}
    \label{eq:like_basic_aux1}
    |e^{-a}-e^{-b}| \leq \exp\left(-\min(a,b)\right) |a-b|.
\end{equation}
Now, noting that $\min\{a,b\} = (a+b)/2 - |a-b|/2$, we have
\begin{align}
    \label{eq:like_basic_aux2}
	\min\big\{-\frac 12 \norm{z_1 - y}_\Gamma^2, & -\frac 12 \norm{z_2 - y}_\Gamma^2\big\} \\
	& = \frac 12 \norm{y}_\Gamma^2 + \min\left\{\frac 12\norm{z_1}_\Gamma^2 - \langle z_1,y\rangle_\Gamma, \frac 12\norm{z_2}_\Gamma^2 - \langle z_2,y\rangle_\Gamma  \right\} \nonumber\\
	& = \frac 12 \norm{y}_\Gamma^2 + \frac 14 \left(\norm{z_1}_\Gamma^2 + \norm{z_2}_\Gamma^2 - 2 \langle z_1 + z_2, y\rangle_\Gamma\right) \nonumber\\
	& \quad - \frac 14 \left|\norm{z_1}_\Gamma^2 - \norm{z_2}_\Gamma^2 - 2 \langle z_1 - z_2, y\rangle_\Gamma\right| \nonumber\\
	& \geq \frac 12 \norm{y}_\Gamma^2 + \frac 14 \norm{z_1}_\Gamma^2 + \frac 14 \norm{z_2}_\Gamma^2 - \frac 12 \left(\frac{\norm{z_1+z_2}_\Gamma^2}{2\tau} + \frac{\tau\norm{y}_\Gamma^2}{2}\right) \nonumber\\
	& \quad -\frac 14 \left|\norm{z_1}_\Gamma^2 - \norm{z_2}_\Gamma^2\right|
	- \frac 12 \left(\frac{\norm{z_1-z_2}_\Gamma^2}{2\tau} + \frac{\tau\norm{y}_\Gamma^2}{2}\right) \nonumber\\
	& \geq \frac{1-\tau}{2} \norm{y}_\Gamma^2 - \frac 1{2\tau}\norm{z_1}_\Gamma^2 - \frac 1{2\tau}\norm{z_2}_\Gamma^2
\end{align}
where we applied the parallelogram identity and the generalized Young's inequality. 
Furthermore, we have
\begin{equation}
\label{eq:like_basic_aux3}
\left|-\frac 12 \norm{z_1 - y}_\Gamma^2 + \frac 12 \norm{z_2 - y}_\Gamma^2\right| 
	\leq  \frac 12 \norm{z_1-z_2}_\Gamma\left(\norm{z_1}_\Gamma + \norm{z_2}_\Gamma + 2\norm{y}_\Gamma \right)    
\end{equation}
Now combining inequalities \eqref{eq:like_basic_aux1}, \eqref{eq:like_basic_aux2} and \eqref{eq:like_basic_aux3} yields the claim.
\end{proof}

\begin{lemma}
\label{lem:Phi_basic}
Let ${\mathcal G}$ satisfy Assumption \ref{ass:prior_perturbation} for a probability measure $\mu$ on ${\mathcal X}$ an suppose $\Phi$ is given by \eqref{eq:Phigauss}. Then it holds that
\begin{equation}
\label{eq:young_joint}
-\frac{1+\tau}{2\tau} \norm{\G(x)}_\Gamma^2 -\frac{1+\tau}2 \norm{y}_\Gamma^2 \leq -\Phi(x;y) \leq -\frac{1-\tau}2 \norm{y}_\Gamma^2 + \frac{1-\tau}{\tau}L_1^2 \norm{x}^2 + C
\end{equation}
for any $\tau>0$, where the constant $C>0$ depends on $\tau$, $R$ and $L_3$. Moreover, $\Phi$ is locally Lipschitz according to
\begin{equation}
    \label{eq:Phi_stab_xxdash}
	|\Phi(x,y) - \Phi(x',y)| \leq L(x,x'; y) \norm{x-x'}
\end{equation}
for all $x,x'\in {\mathcal X}$ and $y\in \R^d$, where $L(x,x'; y) = \frac{L_1}{2} (L_1\norm{x} + L_1\norm{x'} + 2\norm{y}_\Gamma + C).$
Suppose that $\Phi_*$ also satisfies Assumption \ref{ass:prior_perturbation}. Then it holds that
\begin{equation}
	\label{eq:Phi_stab_GGstar}
	|\Phi(x,y) - \Phi_*(x,y)| \leq L_*(x,y) \norm{\G(x) - \G_*(x)}_\Gamma
\end{equation}
for all $x\in X$ and $y\in \R^d$, where $L_*(x,y) = L_1 \norm{x} + \norm{y}_\Gamma + C$. The constant $C$ in the expressions of $L$ and $L_*$ depends on $L_1, L_3$ and $R$.
\end{lemma}

\begin{proof}
By generalized Young's inequality $\langle z,w\rangle \leq \frac{\norm{z}^2}{2\tau} + \frac{\tau}2 \norm{w}^2$ for $\tau>0$, we observe that 
\begin{eqnarray}
	\label{eq:young_neg}
	-\frac 12 \norm{\G(x) - y}_\Gamma^2
	& = & -\frac 12 \norm{y}_\Gamma^2-\frac 12 \norm{\G(x)}_\Gamma^2 + \langle \G(x), y\rangle_\Gamma \nonumber\\
	& \geq & -\frac{1+\tau}{2\tau} \norm{\G(x)}_\Gamma^2 -\frac{1+\tau}2 \norm{y}_\Gamma^2.
\end{eqnarray}
Reversing the Young's inequality we obtain
\begin{eqnarray}
	\label{eq:young_pos}
	-\frac 12 \norm{\G(x) - y}_\Gamma^2
	& = & -\frac 12 \norm{y}_\Gamma^2-\frac 12 \norm{\G(x)}_\Gamma^2 + \langle \G(x), y\rangle_\Gamma \nonumber\\
	& \leq & -\frac{1-\tau}2 \norm{y}_\Gamma^2+\frac{1-\tau}{2\tau} \norm{\G(x)}_\Gamma^2 \nonumber\\
	& \leq & -\frac{1-\tau}2 \norm{y}_\Gamma^2 + \frac{1-\tau}{2\tau}L_1^2 \norm{x}^2 + C.
\end{eqnarray}

Next, the Lipschitz bounds for $\Phi$ in \eqref{eq:Phi_stab_xxdash} and \eqref{eq:Phi_stab_GGstar}
follow by applying \eqref{eq:G_bound} with the inequality \eqref{eq:like_basic_aux3}.
This yields the result.
\end{proof}

The next lemma will be crucial as it provides asymptotic lower and upper bounds for the normalization constant 
\begin{equation}
    \label{eq:Z}
    Z(y) = \E^\mu \exp(-\Phi(x,y))
\end{equation}
with $\Phi$ is given in \eqref{eq:Phigauss}, in the setting specified by Assumption \ref{ass:prior_perturbation}.
Let us also remark that in this case, the normalization constant $Z(y)$ and the evidence density $\pi(y)$ coincide up to universal constant, i.e. $Z(y) = C\pi(y)$, where $C$ depends on $\Gamma$ and the dimension $d$.

\begin{lemma}
\label{lem:Z_lbound}
Let ${\mathcal G}$ satisfy Assumption \ref{ass:prior_perturbation} for a probability measure $\mu$ on ${\mathcal X}$ an suppose $\Phi$ is given by \eqref{eq:Phigauss}. 
For any $\kappa_1>\frac 12$, there exists finite constants $C, C'>0$ such that
\begin{equation}
	\label{eq:Z_equiv}
 C \exp\left(-\kappa_1 \norm{y}_\Gamma^2\right) \leq Z(y) \leq C' \exp\left(-\frac 12\frac {L_2}{L_1^2 + L_2}\norm{y}_\Gamma^2\right)
\end{equation}
for any $y\in\R^d$, where $Z$ is given by \eqref{eq:Z}. The constant $C$ depends on $\kappa_1, L_3, R$ and $d$.
\end{lemma}

\begin{proof}
Applying the lower bound of \eqref{eq:young_joint}, we obtain
\begin{eqnarray*}
	Z(y)
	& \geq & \E^\mu \exp\left(-\frac{1+\tau}{2\tau} \norm{\G(x)}_\Gamma^2\right) \cdot \exp\left(-\frac{1+\tau}2 \norm{y}_\Gamma^2 \right) \\
	& \geq &  \exp\left(-\frac{1+\tau}{2\tau}L_3\right) \mu(B(0,R)) \exp\left(-\frac {1+\tau}2\norm{y}_\Gamma^2\right),
\end{eqnarray*}
Setting $\kappa_1 = (1+\tau)/2>1/2$ yields the lower bound in \eqref{eq:Z_equiv}. 

Similarly, applying the upper bound of \eqref{eq:young_joint} and setting $\tau =\frac{L_1^2}{L_1^2 + L_2}$, we observe
\begin{eqnarray*}
	Z(y) 
	& \leq & \E^\mu \exp\left(\frac {1-\tau}{\tau} L_1^2\norm{x}^2 \right) \cdot \exp\left(-\frac {1-\tau}2\norm{y}_\Gamma^2\right) \\
	& \leq & \E^\mu \exp\left(L_2 \norm{x}^2 \right) \cdot \exp\left(-\frac 12\frac {L_2}{L_1^2 + L_2}\norm{y}_\Gamma^2\right),
\end{eqnarray*}
The upper bound is obtained due to condition $(ii)$ in Assumption \ref{ass:prior_perturbation}.
\end{proof}

\subsection{Proofs for likelihood perturbation}
\label{subsec:like_perturbation}

Throughout this section, we assume that $\G$ and $\G_*$ satisfy Assumption \ref{ass:prior_perturbation} for a probability measure $\mu$ on ${\mathcal X}$ and $\Phi$ and $\Phi_*$ are the corresponding log-likelihoods given by \eqref{eq:Phigauss} for $\G$ and $\G_*$, respectively.
We abbreviate $$\Delta \G = \left(\E^\mu \norm{\G(x) - \G_*(x)}_\Gamma^2\right)^{\frac 12}$$ for convenience as this term will dominate the approximation error in multiple claims.

The next two lemmas characterize the pointwise perturbation of the likelihood energy as well as the normalization constant.

\begin{lemma}
\label{lem:stability:likelihood_est2} 
For any $\kappa>0$, there exists a constant $C>0$ depending on $\kappa$ and $L_1$ such that
\begin{equation*}
	\left|\exp(-\Phi(x,y)) - \exp(-\Phi_*(x,y))\right|\leq C \phi^*_\kappa(y) \psi^*_\kappa(x) \norm{\G(x) - \G_*(x)}_\Gamma
\end{equation*}
for any $x\in X$ and $y\in\R^d$, where
\begin{equation}
	\label{eq:phistarkappa}
	\phi^*_\kappa(y) = \exp\left(\left(\frac{L_1^2}\kappa - \frac 12\right) \norm{y}_\Gamma^2\right)(1+\norm{y}_\Gamma)
\end{equation}
and
\begin{equation}
	\label{eq:psistarkappa}
	\psi^*_\kappa(x) = \exp\left(\kappa\norm{x}^2 \right)(1+\norm{x}).
\end{equation}
\end{lemma}

\begin{proof}
Applying Lemma \ref{lem:basic_quad_zzy} with $z_1 = \G(x)$ and $z_2 = \G_*(x)$ we obtain
\begin{multline*}
    \frac{\left|\exp(-\Phi(x,y)) - \exp(-\Phi_*(x,y))\right|}{\norm{\G(x)-\G_*(x)}_\Gamma} \\
    \leq \frac 12 \left(\norm{\G(x)}_\Gamma + \norm{\G_*(x)}_\Gamma + 2\norm{y}_\Gamma \right)  \exp\left(\frac{\tau-1}{2} \norm{y}_\Gamma^2 + \frac 1{2\tau}\norm{\G(x)}_\Gamma^2 + \frac 1{2\tau}\norm{\G_*(x)}_\Gamma^2\right)  \\
    \leq C \left(L_1 \norm{x} + \norm{y}_\Gamma + C\right)
    \exp\left(\frac{\tau-1}{2} \norm{y}_\Gamma^2 + \frac{2L_1^2}{\tau}\norm{x}^2\right)
\end{multline*}
Choosing $\tau = 2L_1^2/\kappa$ and observing that for $a,b\geq 0$ we have $a+b+1 \leq (a+1)(b+1)$, yields
\begin{equation*}
    \frac{\left|\exp(-\Phi(x,y)) - \exp(-\Phi_*(x,y))\right|}{\norm{\G(x)-\G_*(x)}_\Gamma}
    \leq C(\norm{x}+1)(\norm{y}_\Gamma +1) \exp\left(\left(\frac{L_1^2}\kappa - \frac 12\right)\norm{y}_\Gamma^2 + \kappa \norm{x}^2\right) 
\end{equation*}
This concludes the proof.
\end{proof}

\begin{lemma}
\label{lem:evidence_pert_2}
Let $Z$ and $Z_*$ be the normalization constants defined by \eqref{eq:Z} for $\Phi$ and $\Phi_*$, respectively.
For any $L_2'<L_2$, we have that 
\begin{equation*}
	|Z(y) - Z_*(y)| \leq C \phi_{L_2'/2}^*(y) \Delta \G,
\end{equation*}
where the finite constant $C$ is dependent on $L_2'$ and $L_2$.
\end{lemma}

\begin{proof}
We have that
\begin{eqnarray*}
    |Z(y) - Z_*(y)| & = & \left|\int \left(\exp(-\Phi(x,y))-\exp(-\Phi_*(x,y))\right) \mu(dx)\right| \\
    & \leq & C \int \phi^*_\kappa(y) \psi^*_\kappa(x) \norm{\G(x)-\G_*(x)}_\Gamma \mu(dx) \\
    & \leq & C \phi^*_\kappa(y) \left(\E^\mu \psi^*_\kappa(x)^2\right)^{\frac 12}
    \Delta \G,
\end{eqnarray*}
where we applied the Cauchy-Schwarz inequality.
Now we observe by Corollary \ref{cor:direct_cor_assump} $(ii)$ that $\E^\mu \psi^*_\kappa(x)^2$ is finite for $\kappa = L_2'/2$ with $L_2'<L_2$, which yields the result.
\end{proof}

The following corollary combines the preceding results in a form that will be directly applicable in the proof of Theorem \ref{thm:U1_like_stability}.

\begin{corollary}
\label{cor:like_stab}
Let $L_1^2< \frac{\sqrt 2 - 1}{2} L_2$. Then there exists $L_2'<L_2$ and $\kappa_2>0$ such that
\begin{equation}
    \left|\frac 1{Z(y)}-\frac 1{Z_*(y)}\right| \exp(-\Phi(x,y))
    \leq C \frac{\exp(L_2'\norm{x}^2) \exp(-\kappa_2\norm{y}_\Gamma^2) (1+\norm{y}_\Gamma)}{\max\{Z(y), Z_*(y)\}}\Delta \G,
\end{equation}
where the constant $C>0$ depends on $L_2, L_2', L_3, \kappa_2$ and $d$.
\end{corollary}

\begin{proof}
By applying Lemmas \ref{lem:Phi_basic} (upper bound with $\tau = L_1^2/(L_1^2 + L_2')$), \ref{lem:Z_lbound} (lower bound) and \ref{lem:evidence_pert_2} with $L_2'< L_2$, we obtain
\begin{multline}
    \frac{\left|Z(y)-Z_*(y)\right|}{Z(y)} \exp(-\Phi(x,y)) \\
    \leq C \exp(L_2' \norm{x}^2)\exp\left(\left(\kappa_1 + \frac{2L_1^2}{L_2'}-\frac 12 -\frac 12 \frac {L_2'}{L_1^2 + L_2'}\right)\norm{y}_\Gamma^2\right)(1+\norm{y}_\Gamma) \Delta \G.
\end{multline}
Let us now deduce that admissible parameters $L_2'$ and $\kappa_2$ exist. Observe first that $\kappa_1 - 1/2>0$ can be chosen arbitrarily small and denote
\begin{equation}
    f(t) = 2t - \frac 12 \frac{1}{t+1}.
\end{equation}
In other words, we need to find $L_2'<L_2$ and $\kappa_1>1/2$ so that
\begin{equation}
    -\kappa_2 = \kappa_1 - \frac 12 + f\left(\frac{L_1^2}{L_2'}\right)<0.
\end{equation}
Clearly, this is guaranteed if $f(L_1^2/L_2') < 0$. Now observe that $f$ is increasing for $t>0$ and satisfies $f(0) = -1/2$ and $f((\sqrt 2-1)/2) = 0$. Therefore, $L_1^2/L_2' < (\sqrt 2-1)/2$ can be satisfied if the same inequality holds for $L_2$ as in the assumption.

Since an identical argument applies for $\frac{\left|Z(y)-Z_*(y)\right|}{Z_*(y)} \exp(-\Phi(x,y))$, we obtain the claim.
\end{proof}

Before proceeding, we must first establish that the evidence-averaged moments of the posteriors are bounded in the following way.

\begin{proposition}
\label{prop:like_stab_moment_bound}
Suppose $\G$ and $\G_*$ satisfy Assumption \ref{ass:prior_perturbation} with probability measure $\mu$ on ${\mathcal X}$ and with same constants $L_1, L_2, L_3$ and $R$. Moreover, we assume that $L_1^2< \frac{\sqrt 2 - 1}{2} L_2$.
Then
\begin{equation*}
	\E^{\pi_*(y)} M_p(\mu^y) < \infty \quad \text{and} \quad
	\E^{\pi(y)} M_p(\mu_*^y) < \infty.
\end{equation*}
\end{proposition}

\begin{proof}
Due to symmetry, it is sufficient to prove the first claim. We have that
\begin{eqnarray*}
	\left|\E^{\pi_*(y)} M_p(\mu^y)\right| & \leq & 
	\left|\int \norm{x}^p (h_*(x)-1) \mu(dx)\right| + M_p(\mu)\\
	& \leq & \int \norm{x}^p |h_*(x)-1|\mu(dx)+ M_p(\mu),
\end{eqnarray*}
where
\begin{equation*}
	h_*(x) = \E^{\pi_*(y)} \left[\frac 1{Z(y)} \exp(-\Phi(x,y))\right].
\end{equation*}
Now by Corollary \ref{cor:like_stab} we have that
\begin{eqnarray*}
	|h_*(x)-1| & = & \left|\E^{\pi_*(y)} \left[\left(\frac 1{Z(y)}-\frac 1{Z_*(y)}\right) \exp(-\Phi(x,y))\right]\right| \\
	& \leq & C \exp(L_2'\norm{x}^2) \int \exp(-\kappa_2\norm{y}_\Gamma^2) (1+\norm{y}_\Gamma)dy \cdot \Delta \G,
\end{eqnarray*}
for some $\kappa_2>0$ and $L_2'<L_2$. The result follows due to condition $(ii)$ in Assumption \ref{ass:prior_perturbation}.
\end{proof}

Now we are ready to prove Theorem \ref{thm:U1_like_stability}.
\begin{proof}[Proof of Theorem \ref{thm:U1_like_stability}]
\emph{Claim (i):} Let us first decompose the $W_1$ distance into two error components by writing
\begin{eqnarray*}
	W_1(\mu^y, \mu_*^y)  & =  & \sup_{\phi \in \text{Lip}_1^0({\mathcal X})} \left[\frac 1{Z(y)}\int \phi(x) \exp(-\Phi(x,y)) \mu(dx) - \frac 1{Z_*(y)}\int \phi(x) \exp(-\Phi_*(x,y)) \mu(dx)\right] \\
	& \leq & \frac 1{Z(y)} \sup_{\phi \in \text{Lip}_1^0({\mathcal X})} \left[\int \phi(x) \left\{\exp(-\Phi(x,y)) - \exp(-\Phi_*(x,y))\right\} \mu(dx)\right] \\
	& & +\left|\frac 1{Z(y)} - \frac 1{Z_*(y)}\right|\sup_{\phi \in \text{Lip}_1^0({\mathcal X})} \int \phi(x) \exp(-\Phi_*(x,y)) \mu(dx) \\
	& = & I_1(y) + I_2(y).
\end{eqnarray*}
For the first term, we have by Lemma \ref{lem:stability:likelihood_est2} and the Cauchy-Schwarz inequality that
\begin{eqnarray*}
    Z(y) I_1(y) & = & \sup_{\phi \in \text{Lip}_1^0({\mathcal X})} \left[\int \phi(x) \left\{\exp(-\Phi(x,y)) - \exp(-\Phi_*(x,y))\right\} \mu(dx)\right] \\
    & \leq & C \phi_\kappa^*(y) \E^\mu \left[\norm{x} \psi_\kappa^*(x) \norm{\G(x)-\G_*(x)}_\Gamma \right] \\
    & \leq & C \phi_\kappa^*(y) \left(\E^\mu \norm{x}^2 (\psi_\kappa^*(x))^2\right)^{\frac 12} \Delta \G
\end{eqnarray*}
The term $\E^\mu \norm{x}^2 (\psi_\kappa^*(x))^2$ is bounded for $\kappa < L_2/2$. This yields the bound
\begin{equation*}
    Z(y) I_1(y) \leq C \phi_{L_2'/2}^*(y) \Delta \G,
\end{equation*}
where $L_2'<L_2$ is arbitrary.

By Corollary \ref{cor:like_stab}, the second term satisfies 
\begin{equation*}
    I_2(y) \leq C \frac{\exp(-\kappa_2\norm{y}_\Gamma^2) (1+\norm{y}_\Gamma)}{\max\{Z(y), Z_*(y)\}} \E^\mu\left(\norm{x} \exp(L_2'\norm{x}^2)\right) \Delta \G,
\end{equation*}
where the exponential moment is bounded since $L_2'<L_2$.

In consequence, we have
\begin{equation*}
    \E^{\pi(y)} \left(I_1(y) + I_2(y)\right)
    \leq 
    C \int \left(\exp\left(\left(2\frac{L_1^2}{L_2'}- \frac 12\right) \norm{y}_\Gamma^2 + \exp(-\kappa_2\norm{y}_\Gamma^2)\right) (1+\norm{y}_\Gamma)\right) dy \Delta \G,
\end{equation*}
where the integral convergences since $L_2'$ can be chosen arbitrarily close to $L_2$ and
\begin{equation*}
    2\frac{L_1^2}{L_2}- \frac 12 < \sqrt 2 - 1 - \frac 12 < 0.
\end{equation*}
This concludes the claim.

\emph{Claim (ii):} We have by triangle inequality that
\begin{eqnarray*}
        \vert U_1 - U_1^*\vert & \leq & \left\vert \E^{\pi(y)}\left(W_1(\mu,\mu_*^{y}) - W_1(\mu,\mu_*^{y})\right)\right\vert\\
        & & + \left\vert \E^{\pi(y)}W_1(\mu,\mu_*^{y}) - \E^{\pi_*(y)}W_1(\mu,\mu_*^{y})\right\vert \\
        & \leq &  \E^{\pi(y)} W_1(\mu^y, \mu_*^y) + \left|(\E^{\pi(y)} - \E^{\pi_*(y)})W_1(\mu,\mu_*^{y})\right|.
\end{eqnarray*}
The first term is bounded directly by the posterior bound in $(i)$. For the second term, we observe
\begin{eqnarray}
    \label{eq:evidence_diff_like1}
	\left\vert \left[\E^{\pi(y)} - \E^{\pi_*(y)}\right] W_1(\mu,\mu_*^{y})\right\vert
    & = & \left|\int  W_1(\mu,\mu_*^{y}) \int (\pi(y | x) - \pi_*(y | x)) \mu(dx) dy \right|  \nonumber  \\
    & \leq & C \int W_1(\mu,\mu_*^{y})\left| \int \left(\exp(-\Phi(x,y)) - \exp(-\Phi_*(x,y)) \right)\mu(dx)  \right| dy \nonumber \\
    & \leq & C \int W_1(\mu,\mu_*^{y})\phi_{L_2'/2}^*(y) dy \cdot \int  \psi_{L_2'/2}^*(x) \norm{\G(x) - \G_*(x)}_\Gamma \mu(dx) \nonumber \\
    & \leq & C \int W_1(\mu,\mu_*^{y})\phi_{L_2'/2}^*(y) dy \cdot \left(\E^\mu \psi_{L_2'/2}^*(x)^2\right)^{\frac 12} \Delta \G,
\end{eqnarray}
where we applied Lemma \ref{lem:stability:likelihood_est2} and the Cauchy-Schwarz inequality.

Finally, we observe by Proposition \ref{prop:general_properties_Wp} claim (3) that $W_1(\mu,\mu_*^{y}) \leq M_1(\mu) + M_1(\mu_*^y)$ and
\begin{equation*}
    M_1(\mu_*^y) \phi_{L_2'/2}^*(y) \leq \E^\mu \left(\norm{x} \exp(L_2' \norm{x}^2)\right)
    \exp\left(\left(\kappa_1 - \frac 12 \frac{L_2'}{L_1^2+L_2
    } + 2\frac{L_1^2}{L_2'} - \frac 12\right)\norm{y}_\Gamma^2\right),
\end{equation*}
where we applied Lemmas \ref{lem:Phi_basic} (upper bound in \eqref{eq:young_joint} and \ref{lem:Z_lbound} (lower bound), and for $L_2'<L_2$ set $\tau = L_1^2/(L_1^2+L_2')$.
Since $\kappa_1>1/2$ can be set arbitrarily close to $1/2$, we have
\begin{equation*}
    \kappa_1 - \frac 12 \frac{L_2'}{L_1^2+L_2} + 2\frac{L_1^2}{L_2'} - \frac 12
    = \kappa_1 - \frac 12 + f\left(\frac{L_1^2}{L_2'}\right)<0
\end{equation*}
for some $L_2'<L_2$ as $f(L_1^2/L_2)<0$ by assumption. In consequence, the integral over $y$ in \eqref{eq:evidence_diff_like1} is finite, which proves the claim.
This concludes the proof.
\end{proof}

\subsection{Proofs for prior stability}
\label{subsec:prior_perturbation}

We now examine the prior stability of $U_p$ for the cases $p=1$ and $p=2$. To support this analysis, we first derive auxiliary results that will be applied in both cases. Throughout this section, we assume that $\G$ satisfies Assumption \ref{ass:prior_perturbation} for two probability measures $\mu$ and $\tmu$ on ${\mathcal X}$. Moreover, $\Phi$ is the corresponding log-likelihoods given by \eqref{eq:Phigauss} for $\G$.

In terms of the proof strategy, we follow a similar template for auxiliary results as with the likelihood perturbation before proceeding to the main proof.

\begin{lemma}
\label{lem:stability:likelihood_est}
For any $\kappa>0$, there exists a constant $C>0$ depending on $\kappa$ and $L_1$ such that
\begin{equation*}
	\left|\exp(-\Phi(x,y)) - \exp(-\Phi(x',y))\right|\leq C \phi_\kappa(y) \psi_\kappa(x) \psi_\kappa(x')\norm{x-x'}
\end{equation*}
for any $x,x'\in X$ and $y\in\R^d$, where
\begin{equation}
	\label{eq:phikappa}
	\phi_\kappa(y) = \exp\left( \frac{1}{2}\frac{L_1^2-\kappa}{\kappa} \norm{y}_\Gamma^2\right)(1+\norm{y}_\Gamma)
\end{equation}
and
\begin{equation}
	\label{eq:psikappa}
	\psi_\kappa(x) = \exp\left(\kappa\norm{x}^2 \right)(1+\norm{x}).
\end{equation}
\end{lemma}

\begin{proof}
Applying Lemma \ref{lem:basic_quad_zzy} with $z_1 = \G(x)$ and $z_2 = \G(x')$ we obtain
\begin{multline*}
    \frac{\left|\exp(-\Phi(x,y)) - \exp(-\Phi(x',y))\right|}{\norm{\G(x)-\G(x')}_\Gamma} \\
    \leq \frac 12 \left(\norm{\G(x)}_\Gamma + \norm{\G(x')}_\Gamma + 2\norm{y}_\Gamma \right)  \exp\left(\frac{\tau-1}{2} \norm{y}_\Gamma^2 + \frac 1{2\tau}\norm{\G(x)}_\Gamma^2 + \frac 1{2\tau}\norm{\G(x')}_\Gamma^2\right)  \\
    \leq C \left(\norm{x} + \norm{x'} + \norm{y}_\Gamma + 1\right)
    \exp\left(\frac{\tau-1}{2} \norm{y}_\Gamma^2 + \frac{L_1^2}{\tau}\left(\norm{x}^2+\norm{x'}^2\right)\right)
\end{multline*}
Choosing $\tau = L_1^2/\kappa$ and observing that for $a,b\geq 0$ we have $a+b+1 \leq (a+1)(b+1)$, yields
\begin{equation*}
    \frac{\left|\exp(-\Phi(x,y)) - \exp(-\Phi(x',y))\right|}{\norm{\G(x)-\G(x')}_\Gamma}
    \leq C(\norm{x}+\norm{x'}+1)(\norm{y}_\Gamma +1) \exp\left(\frac{1}{2}\frac{L_1^2-\kappa}{\kappa}\norm{y}_\Gamma^2 + \kappa \norm{x}^2 + \kappa \norm{x'}^2\right) 
\end{equation*}
This yields the claim.
\end{proof}

\begin{lemma}
\label{lem:evidence_pert}
Let $Z$ and $\widetilde{Z}$ be the normalization constants defined by \eqref{eq:Z} for $\mu$ and $\tilde \mu$, respectively.
For any $L_2'<L_2$, we have that 
\begin{equation*}
	|Z(y) - \widetilde{Z}(y)|  \leq C \exp\left( \left(\frac{L_1^2}{L_2'}-\frac 12\right) \norm{y}_\Gamma^2\right)(1+\norm{y}_\Gamma) \cdot W_2(\mu,\tmu).
\end{equation*}
\end{lemma}

\begin{proof}
Let $\rho$ be any coupling between $\mu$ and $\tmu$. We have
\begin{eqnarray*}
	|Z(y) - \widetilde{Z}(y)| & =  & \left|\int (\exp(-\Phi(x,y)) - \exp(-\Phi(x',y)) \rho(dx,dx') \right| \\
	& \leq & C \phi_\kappa(y) \int \psi_\kappa(x) \psi_\kappa(x') \norm{x-x'} \rho(dx,dx') \\
	& \leq & C \phi_\kappa(y) \cdot \sqrt{\E^\mu\left[\psi_\kappa(x)^2\right]} \cdot \sqrt{\E^{\tmu}\left[\psi_\kappa(x)^2\right]} \cdot \sqrt{\int \norm{x-x'}^2 \rho(dx, dx')}.
\end{eqnarray*}
The choice $\kappa = L_2'/2$ with $L_2'<L_2$ guarantees the boundedness of the exponential moments and yields
\begin{equation*}
	|Z(y) - \widetilde{Z}(y)| \leq C \exp\left( \left(\frac{L_1^2}{L_2'}-\frac 12\right) \norm{y}_\Gamma^2\right)(1+\norm{y}_\Gamma) \cdot \sqrt{\int \norm{x-x'}^2 \rho(dx, dx')}.
\end{equation*}
Since the coupling $\rho$ was arbitrary, this proves the claim.
\end{proof}

\begin{corollary}
\label{cor:prior_stab_aux1}
Let $L_1^2< \frac{\sqrt 3 -1}2 L_2$. There exists $L_2'<L_2$ and $\kappa_2>0$ such that
\begin{equation}
    \left|\frac 1{Z(y)}-\frac 1{\widetilde{Z}(y)}\right| \exp(-\Phi(x,y))
    \leq C \frac{\exp(L_2'\norm{x}^2) \exp(-\kappa_2\norm{y}_\Gamma^2) (1+\norm{y}_\Gamma)}{\max\{Z(y), \widetilde{Z}(y)\}}W_2(\mu,\tmu) .
\end{equation}
\end{corollary}

\begin{proof}
By applying Lemmas \ref{lem:Phi_basic} (upper bound with $\tau = L_1^2/(L_1^2 + L_2')$), \ref{lem:Z_lbound} (lower bound) and \ref{lem:evidence_pert} with $L_2'< L_2$, we obtain
\begin{multline}
    \frac{\left|Z(y)-\widetilde{Z}(y)\right|}{Z(y)} \exp(-\Phi(x,y)) \\
    \leq C \exp(L_2' \norm{x}^2)\exp\left(\left(\kappa_1 + \frac{L_1^2}{L_2'}-\frac 12-\frac 12 \frac {L_2'}{L_1^2 + L_2'}\right)\norm{y}_\Gamma^2\right)(1+\norm{y}_\Gamma) W_2(\mu,\tmu).
\end{multline}
Following the footsteps in the proof of Corollary \ref{cor:like_stab}, we observe that $\kappa_1 - 1/2>0$ can be chosen arbitrarily small and denote
\begin{equation}
    g(t) = t - \frac 1{t+1}.
\end{equation}
We note that $g$ is increasing for $t>0$ with $g(0)=-1/2$ and $g((\sqrt 3 -1)/2)) = 0$. Therefore, our assumption $L_1^2/L_2< \frac{\sqrt 3 -1}2$ guarantees that $\kappa_1$ and $L_2'<L_2$ can be chosen such that
\begin{equation*}
    -\kappa_2 = \kappa_1 - \frac 12 + g\left(\frac{L_1^2}{L_2'}\right) < 0.
\end{equation*}
As the same upper bound holds for $\frac{\left|Z(y)-\widetilde{Z}(y)\right|}{\widetilde{Z}(y)} \exp(-\Phi(x,y))$, we obtain the result.
\end{proof}

\begin{proposition}
\label{prop:stab_moment_bound}
Suppose $\G$ and $\G_*$ satisfy Assumption \ref{ass:prior_perturbation} with probability measure $\mu$ on ${\mathcal X}$ and with same constants $L_1, L_2, L_3$ and $R$. Moreover, we assume $L_1^2<\frac{\sqrt 3 -1}2 L_2$.
Then
\begin{equation*}
	\E^{\tpi(y)} M_p(\mu^y) < \infty \quad \text{and} \quad
	\E^{\pi(y)} M_p(\tmu^y) < \infty.
\end{equation*}
\end{proposition}

\begin{proof}
Here, we follow reiterate the proof of Proposition \ref{prop:stab_moment_bound}. Again, due to symmetry, it is sufficient to prove the first claim. We have that
\begin{equation*}
	\left|\E^{\tpi(y)} M_p(\mu^y)\right| \leq  \int \norm{x}^p |\tilde h(x)-1|\tmu(dx)+ M_p(\tmu),
\end{equation*}
where
\begin{equation*}
	\tilde h(x) = \E^{\tpi(y)} \left[\frac 1{Z(y)} \exp(-\Phi(x,y))\right].
\end{equation*}
Now by Corollary \ref{cor:prior_stab_aux1} we have that
\begin{eqnarray*}
	|\tilde h(x)-1| & = & \left|\E^{\tpi(y)} \left[\left(\frac 1{Z(y)}-\frac 1{\widetilde{Z}(y)}\right) \exp(-\Phi(x,y))\right]\right| \\
	& \leq & C \exp\left(L_2' \norm{x}^2\right) \int \exp\left(-\kappa_2\norm{y}_\Gamma^2\right)(1+\norm{y}_\Gamma)dy \cdot W_2(\mu,\tmu),
\end{eqnarray*}
for some $\kappa_2>0$ and $L_2'<L_2$. The result follows due to condition $(ii)$ in Assumption \ref{ass:prior_perturbation}.
\end{proof}

\subsubsection{Prior perturbations for $p=1$}

Recall that by Kantorovich duality we can write
\begin{equation*}
	W_1(\mu_1, \mu_2)  = \sup_{\phi \in \text{Lip}_1^0({\mathcal X})} \int \phi(x) \left[\mu_1(dx) - \mu_2(dx)\right],
\end{equation*}
where $\text{Lip}_1^0({\mathcal X})$ stands for the Lipschitz-$1$ (i.e. Lipschitz constant is bounded by $1$) that vanish at the origin.  Below, we use in particular that $|\phi(x)| = |\phi(x) - \phi(0)| \leq \norm{x}$ for any $\phi \in \text{Lip}_1^0({\mathcal X})$.

Next theorem is based on the ideas in \cite[Thm. 3.8]{garbuno2023bayesian} and is modified for the purpose of this section.

\begin{proof}[Proof of Theorem \ref{thm:U1_prior_stability}]
\emph{Claim $(i)$:}
Let us first decompose the $W_1$ distance into two error components by writing
\begin{eqnarray*}
	W_1(\mu^y, \tmu^y)  & =  & \sup_{\phi \in \text{Lip}_1^0({\mathcal X})} \left[\frac 1{Z(y)}\int \phi(x) \exp(-\Phi(x,y)) \mu(dx) - \frac 1{\widetilde{Z}(y)}\int \phi(x) \exp(-\Phi(x,y)) \tmu(dx)\right] \\
	& \leq & \frac 1{Z(y)} \sup_{\phi \in \text{Lip}_1^0({\mathcal X})} \left[\int \phi(x) \exp(-\Phi(x,y)) \mu(dx) - \int \phi(x) \exp(-\Phi(x,y)) \tmu(dx)\right] \\
	& & +\left|\frac 1{Z(y)} - \frac 1{\widetilde{Z}(y)}\right|\sup_{\phi \in \text{Lip}_1^0({\mathcal X})} \int \phi(x) \exp(-\Phi(x,y)) \tmu(dx) \\
	& = & I_1(y) + I_2(y).
\end{eqnarray*}
Considering the error first term, let $\rho$ be a coupling of $\mu$ and $\tmu$. Now we obtain
\begin{eqnarray*}
		Z(y) I_1(y) & = & \sup_{\phi \in \text{Lip}_1^0({\mathcal X})} \int \left[\phi(x)\exp(-\Phi(x,y)) - \phi(x')\exp(-\Phi(x',y))\right] \rho(dx,dx') \\
        & \leq &  \sup_{\phi \in \text{Lip}_1^0({\mathcal X})} \int |\phi(x)| \left|\exp(-\Phi(x,y)) - \exp(-\Phi(x',y))\right| \rho(dx,dx')\\
        & & + \sup_{\phi \in \text{Lip}_1^0({\mathcal X})} \int  \exp(-\Phi(x,y)) |\phi(x) - \phi(x')| \rho(dx,dx') \\
	& \leq & \int \left(\norm{x} |\exp(-\Phi(x,y)) - \exp(-\Phi(x',y))| + \exp(-\Phi(x,y)) \norm{x-x'}\right) \rho(dx,dx') \\
    & \leq & C \int \bigg( \norm{x} \phi_\kappa(y) \psi_\kappa(x) \psi_\kappa(x') 
	 \\
    & & \quad\quad\quad + \exp\left(\frac{1-\tau}{\tau} L_1^2\norm{x}^2\right) \exp\left(-\frac{1-\tau}2 \norm{y}_\Gamma^2\right) \bigg)\norm{x-x'} \rho(dx,dx')
\end{eqnarray*}
for $\kappa, \tau>0$, where we utilized Lemmas \ref{lem:stability:likelihood_est} and \ref{lem:Phi_basic}. Setting $\kappa = L_2'/2$ with $L_2'<L_2$ and $\tau = 2L_1^2/(2L_1^2+L_2)$, and applying Cauchy-Schwarz inequality, and taking into account that $\rho$ was arbitrary, we obtain
\begin{multline*}
    Z(y) I_1(y) \\
    \leq C \exp\left(\frac 12 \frac{2L_1^2 - L_2'}{L_2'}\norm{y}_\Gamma^2\right)(1+\norm{y}_\Gamma) \left(\E^\mu(\norm{x}^2 \psi_{L_2'/2}(x)^2)\right)^{\frac 12}\left(\E^{\tmu} \psi_{L_2'/2}(x)^2\right)^{\frac 12} W_2(\mu, \tmu) \\
    + \exp\left(\left(-\frac 12 + \frac{L_1^2}{2L_1^2 + L_2'}\right)\norm{y}_\Gamma^2\right)\left(\E^\mu \exp(L_2' \norm{x}^2)\right)^{\frac 12} W_2(\mu,\tmu)
\end{multline*}
Clearly, our assumption on $L_1$ and $L_2$ implies $L_1^2 < L_2'/2$ and, consequently, $(2L_1^2-L_2')/L_2'<0$. Moreover, for any $L_2'>0$ we have $-1/2 + L_1^2/(2L_1^2 + L_2') < 0$, guaranteeing the exponential decay of $Z(y)I_1(y)$ and, therefore, finite expectation
\begin{equation}
    \E^{\pi(y)} I_1(y) \leq C W_2(\mu,\tmu).
\end{equation}

Consider now the second error term $I_2$. We first note that by Corollary \ref{cor:prior_stab_aux1} there exists $L_2'<L_2$ and $\kappa_2>0$ such that
\begin{eqnarray*}
    I_2(y) & \leq &  \int \norm{x} \left|\frac 1{Z(y)} - \frac 1{\widetilde{Z}(y)}\right|\exp(-\Phi(x,y)) \tmu(dx) \\
    & \leq & C \E^{\tmu} \left(\norm{x} \exp(L_2'\norm{x}^2)\right) \frac{\exp(-\kappa_2\norm{y}_\Gamma^2) (1+\norm{y}_\Gamma)}{\max\{Z(y), \widetilde{Z}(y)\}} \cdot W_2(\mu,\tmu)
\end{eqnarray*}
Now we observe that the expectation over $\pi(y)$ is bounded, which concludes the proof.

\emph{Claim $(ii)$:}
We have by triangle inequality and properties of the Wasserstein distance that
\begin{eqnarray}
    \label{eq:U1prior_stability_proof_ineq1}
        \vert U_1 - \widetilde U_1\vert & \leq & \left\vert \E^{\pi(y)}W_1(\mu,\mu^{y}) - \E^{\pi(y)}W_1(\tmu,\mu^{y})\right\vert\\
        && + \left\vert \E^{\pi(y)}W_1(\tmu,\mu^{y}) - \E^{\pi(y)}W_1(\tmu,\tmu^{y})\right\vert \nonumber\\
        & & + \left\vert \E^{\pi(y)}W_1(\tmu,\tmu^{y}) - \E^{\tpi(y)}W_1(\tmu,\tmu^{y})\right\vert \nonumber \\
        & \leq &  W_1(\mu, \tmu) + \E^{\pi(y)} W_1(\mu^y, \tmu^y) + \left|(\E^{\pi(y)} - \E^{\tpi(y)})W_1(\tmu,\tmu^{y})\right|
\end{eqnarray}
By the first claim, it holds that $\E^{\pi(y)} W_1(\mu^y, \tmu^y) \leq C W_2(\mu,\tmu)$. 

Consider now the third term and let $\rho$ be any coupling between $\mu$ and $\tmu$. By Cauchy-Schwarz and Lemma \ref{lem:stability:likelihood_est}, we obtain
\begin{align}
    \label{eq:U1prior_stability_proof_ineq3}
	\bigg\vert \big[\E^{\pi(y)} & - \E^{\tpi(y)}\big] W_1(\tmu,\tmu^{y})\bigg\vert \nonumber\\
	& = \left|\int  W_1(\tmu,\tmu^{y})\int (\pi(y | x) - \pi(y | x')) \rho(dx,dx')   dy\right| \nonumber \\
	& \leq \int W_1(\tmu,\tmu^{y}) \int \frac{|\pi(y | x) - \pi(y | x')|}{\norm{x-x'}} \cdot \norm{x-x'} \rho(dx,dx')  dy \nonumber \\
	& \leq \int W_1(\tmu,\tmu^{y})  \left(\int \frac{|\pi(y | x) - \pi(y | x')|^2}{\norm{x-x'}^2}\rho(dx,dx')\right)^{\frac 12} dy\left(\int \norm{x-x'}^2 \rho(dx,dx')\right)^{\frac 12} \nonumber \\
	& \leq \int W_1(\tmu,\tmu^{y}) \phi_{L_2'/2}(y) dy \cdot (\E^\mu \psi_{L_2'/2}(x)^2)^{\frac 12} (\E^{\tmu} \psi_{L_2'/2}(x)^2)^{\frac 12} \left(\int \norm{x-x'}^2 \rho(dx,dx')\right)^{\frac 12},
\end{align}
where $\phi_\kappa$ and $\psi_\kappa$ are given by \eqref{eq:phikappa} and \eqref{eq:psikappa}, respectively.
Similar to the proof of Theorem \ref{thm:U1_like_stability} $(ii)$, we observe by Proposition \ref{prop:general_properties_Wp} claim (3) that $W_1(\tmu, \tmu^y) \leq M_1(\tmu) + M_1(\tmu^y)$ and
\begin{equation}
    \label{eq:U1priorstability_proof_ineq4}
    M_1(\tmu^y) \phi_{L_2'/2}(y) \leq \E^\mu \left(\norm{x} \exp(L_2' \norm{x}^2)\right)
    \exp\left(\left(\kappa_1 - \frac 12 \frac{L_2'}{L_1^2+L_2
    } + \frac{L_1^2}{L_2'} - \frac 12\right)\norm{y}_\Gamma^2\right),
\end{equation}
where we applied Lemmas \ref{lem:Phi_basic} (upper bound in \eqref{eq:young_joint} and \ref{lem:Z_lbound} (lower bound), and for $L_2'<L_2$ set $\tau = L_1^2/(L_1^2+L_2')$. Since $\kappa_1>1/2$ can be set arbitrarily close to $1/2$, we have
\begin{equation*}
    \kappa_1 - \frac 12 \frac{L_2'}{L_1^2+L_2
    } + \frac{L_1^2}{L_2'} - \frac 12 = \kappa_1 - \frac 12 + g\left(\frac{L_1^2}{L_2'}\right) < 0
\end{equation*}
for some $L_2'<L_2$ as $g(L_1^2/L_2) < 0$ by assumption.

Finally, since the coupling $\rho$ was arbitrary, it follows that
\begin{equation}
    \label{eq:U1prior_stability_proof_ineq2}
	\left\vert \left[\E^{\pi(y)} - \E^{\tpi(y)}\right] W_1(\tmu,\tmu^{y})\right\vert \leq C W_2(\mu, \tmu).
\end{equation}
Combining inequality \eqref{eq:U1prior_stability_proof_ineq2} with \eqref{eq:U1prior_stability_proof_ineq1} yields the claim.
\end{proof}

\subsubsection{Prior perturbations for $p=2$}

We follow the same proof strategy for Theorem \ref{thm:stability:main} as for Theorem \ref{thm:U1_prior_stability}, with only minor modifications.

\begin{proof}[Proof of Theorem \ref{thm:stability:main}]
Let us decompose the error term into three terms
\begin{eqnarray*}
        \vert U_2 - \widetilde U_2\vert & \leq & \left\vert \E^{\pi(y)}W_{2}^{2}(\mu,\mu^{y}) - \E^{\pi(y)}W_{2}^{2}(\tmu,\mu^{y})\right\vert\\
        && + \left\vert \E^{\pi(y)}W_{2}^{2}(\tmu,\mu^{y}) - \E^{\pi(y)}W_{2}^{2}(\tmu,\tmu^{y})\right\vert\\
        & & + \left\vert \E^{\pi(y)}W_{2}^{2}(\tmu,\tmu^{y}) - \E^{\tpi(y)}W_{2}^{2}(\tmu,\tmu^{y})\right\vert\\
& \leq & \E^{\pi(y)}\left|W_{2}(\mu,\mu^{y}) + W_{2}(\tmu,\mu^{y})\right|\left|W_{2}(\mu,\mu^{y}) - W_{2}(\tmu,\mu^{y})\right|\\
        && + \E^{\pi(y)}\left|W_{2}(\tmu,\mu^{y}) + W_{2}(\tmu,\tmu^{y})\right|\left|W_{2}(\tmu,\mu^{y}) - W_{2}(\tmu,\tmu^{y})\right|\\
        & & + \left\vert \left[\E^{\pi(y)} - \E^{\tpi(y)}\right] W_2^2(\tmu,\tmu^{y})\right\vert\\        
        &\leq & \widetilde{K}_1 W_2(\mu,\tmu) + K_2 \sqrt{\E^{\pi(y)} W_2^2(\mu^y, \tmu^y)} + \left\vert \left[\E^{\pi(y)} - \E^{\tpi(y)}\right] W_2^2(\tmu,\tmu^{y})\right\vert,
\end{eqnarray*}
where we applied the triangle inequality and the Cauchy-Schwartz inequality in first and second terms, respectively. Moreover, by Proposition \ref{prop:general_properties_Wp} the constants $K_1$ and $K_2$ satisfy
\begin{eqnarray*}
	\widetilde{K}_1 & \leq & \sqrt 2 \E^{\pi(y)} \left[\sqrt{M_2(\mu)+M_2(\mu^y)}+\sqrt{M_2(\tmu)+M_2(\mu^y)}\right] \\
	& \leq & 2 \sqrt{M_2(\mu)} + \sqrt{2} \sqrt{M_2(\tmu)+M_2(\mu)} < \infty
\end{eqnarray*}
and
\begin{eqnarray*}
	K_2 & \leq & \sqrt{\E^{\pi(y)} (W_{2}(\tmu,\mu^{y}) + W_{2}(\tmu,\tmu^{y}))^2 } \\
	& \leq & 2 \sqrt{\E^{\pi(y)} (2 M_2(\tmu) + M_2(\mu^y) + M_2(\tmu^y)} < \infty
\end{eqnarray*}
due to Proposition \ref{prop:stab_moment_bound}.

Consider now the third term. Following the same deduction as in inequality \eqref{eq:U1prior_stability_proof_ineq3} we have
\begin{multline*}
	\left\vert \left[\E^{\pi(y)} - \E^{\tpi(y)}\right] W_2^2(\tmu,\tmu^{y})\right\vert \\
	\leq \int W_2^2(\tmu,\tmu^{y}) \phi_{L_2'/2}(y) dy \cdot (\E^\mu \psi_{L_2'/2}(x)^2)^{\frac 12} (\E^{\tmu} \psi_{L_2'/2}(x)^2)^{\frac 12} \left(\int \norm{x-x'}^2 \rho(dx,dx')\right)^{\frac 12},
\end{multline*}
By Proposition \ref{prop:general_properties_Wp} claim (3) that $W_2(\tmu, \tmu^y) \leq 2(M_2(\tmu) + M_2(\tmu^y))$ and
similar to inequality \ref{eq:U1priorstability_proof_ineq4} we have
\begin{equation*}
    M_2(\tmu^y) \phi_{L_2'/2}(y) \leq \E^\mu \left(\norm{x}^2 \exp(L_2' \norm{x}^2)\right)
    \exp\left(\left(\kappa_1 - \frac 12 \frac{L_2'}{L_1^2+L_2
    } + \frac{L_1^2}{L_2'} - \frac 12\right)\norm{y}_\Gamma^2\right),
\end{equation*}
where the exponent is negative for some $\kappa_1>1/2$ and $L_2'<L_2$. In consequence, we have
\begin{equation*}
    \left\vert \left[\E^{\pi(y)} - \E^{\tpi(y)}\right] W_2^2(\tmu,\tmu^{y})\right\vert
    \leq C W_2(\mu, \tmu).
\end{equation*}
Combining the arguments above, we obtain the claim.
\end{proof}

\section{Algorithms and simulations}
\label{sec:algorithms}

In this section, we demonstrate computability of the Wasserstein criterion and the predicted numerical rates through simplified examples. Our computations below focus mainly on the $p=2$ case and connections to optimal transport (see e.g., \cite{villani2008optimal}), but we also demonstrate the predicted convergence rates of empirical measure approximations in the Wasserstein-$1$ distance.

\subsection{Prior stability and empirical measures}

Let us consider approximations of a prior measure $\mu$ by an empirical measure $\mu_{M} = \frac{1}{M}\sum_{m=1}^{M}\delta(x-x^m)$, where $x^m \sim \mu$ i.i.d. In such a case, the posterior measure follows the formula
\begin{equation}
    \mu_M^y = \frac 1M\sum_{m=1}^{M}w_{m}^{y}\delta(x-x^m),
\end{equation}
where 
\begin{equation*}
    w_m^y = \frac{1}{Z_{M}(y)}\exp(-\Phi(x^m,y))\quad \text{and} \quad    
    Z_{M}(y)=\frac{1}{M}\sum_{k=1}^M\exp(-\Phi(x^{k},y))
\end{equation*}
Now, suppose our observation emerges from an inverse problem \eqref{eq:BIP} with $y \sim {\mathcal N}\left({\mathcal G}(x; \theta), \Gamma\right)$. Given the prior $\mu_M$, the evidence follows a Gaussian mixture model, and the expected utility for the approximate model satisfies
\begin{equation*}
    U_1^M = \E^{\pi_M} W_1(\mu_M, \mu_M^y) = \frac{1}{M}\sum_{m=1}^{M}\mathbb{E}^{\mathcal{N}({\mathcal G}(x^m),\Gamma)} 
    W_1(\mu_{M},\mu_{M}^{y}).
\end{equation*}
For general computational perspective, we note that the Wasserstein-$1$ distance between discrete measures is reduced to a linear programming task (see e.g., \cite{peyre2020computational}). 

To demonstrate the approximation rates of the expected utility in Theorem \ref{thm:U1_prior_stability}, let us consider a one-dimensional toy example, where ${\mathcal G} : \R\times [-1,1] \to \R$, where ${\mathcal G}(x;\theta) = 5 \theta^{6} x$. Moreover, we assume a normal prior distribution $\mu = \mathcal{N}(0,1)$. With normally distributed additive measurement noise $\epsilon \sim {\mathcal N}(0,0.05^2)$, the posterior also has Gaussian statistics and enables straightforward means to evaluate the exact expected utility.

It is well-known (see e.g. Corollaries 6.10 and 6.14 in \cite{bobkov2019one}) that empirical approximations of one-dimensional Gaussian distributions satisfy
\begin{equation}
    \E^{\otimes \mu} W_1(\mu, \mu_M) \lesssim \frac 1{\sqrt M} \quad \text{and} \quad
    \E^{\otimes \mu} W_2(\mu, \mu_M) \lesssim \sqrt{\frac{\log \log M}{M}},
\end{equation}
where $\E^{\otimes \mu}$ stands for the ensemble average and the proportionality constants are universal. In light of Theorem \ref{thm:U1_prior_stability}, we now expect to observe
\begin{equation}
    \label{eq:U1_expected_rate_for_empirical}
    \E^{\otimes \mu}|U_1(\theta) - U_1^M(\theta)| \lesssim \sqrt{\frac{\log \log M}{M}},
\end{equation}
where $U_1$ and $U_1^M$ correspond to the true and approximated expected utility, respectively.

Let us now briefly outline the evaluation of the two expected utilities. We utilize identities $W_{1}(\mu_1,\mu_2) = \int_{\mathbb{R}}\vert F_{1}(x) - F_{2}(x) \vert dx$ and $W_{1}(\mu_1,\mu_2) = \int_{[0,1]}\vert F_{1}^{-1}(x) - F_{2}^{-1}(x) \vert dx$,
where $F_{i}$ is the cumulative distribution function of $\mu_i$. 
First, the inverse cumulative distributions of the Gaussian prior and posterior can be expressed in terms of the inverse error function $\text{erf}^{-1}$, allowing direct computation of the $W_1$ distance. Second, for empirical measures, the cumulative distributions can be replaced with their empirical counterparts, enabling efficient evaluation via the first formula. The expectation over the corresponding evidence distributions is estimated as a combination of Gaussian integrals. Here, each integral was approximated using Gauss-Hermite Smolyak quadratures, with 33 nodes (see e.g., \cite{le2010spectral,xiu2010numerical}).

The true expected utility and the numerical convergence is plotted in Figure \ref{fig:U1_combined}.
We pick 3 design values $\theta \in \lbrace A,B,C\rbrace$ and estimate the ensemble average $\E^{\otimes \mu}\vert U_1(\theta) - U_1^M(\theta) \vert$ in each node over varying $M$ ranging from 10 to 39810 with ensemble sizes of 200 samples. We observe convergence rates approximately proportional to $1/\sqrt{M}$.

\begin{figure}[htp]
\centering
\begin{subfigure}[b]{0.48\textwidth}
    \centering
    \includegraphics[width=\textwidth]{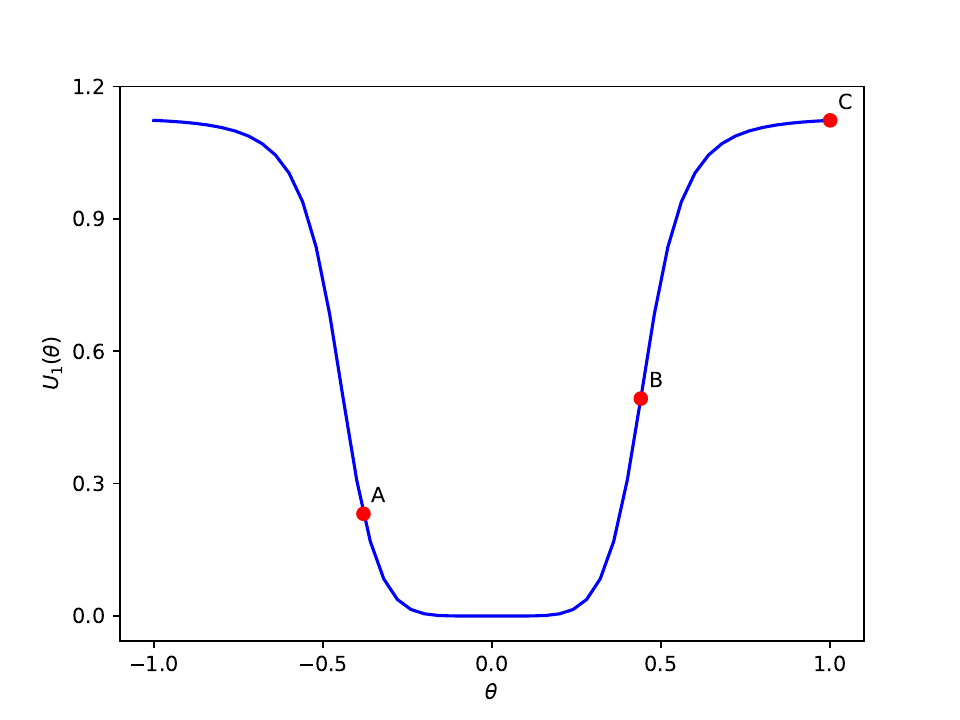}
    \label{U1_nodes}
\end{subfigure}
\hfill
\begin{subfigure}[b]{0.48\textwidth}
    \centering
    \includegraphics[width=\textwidth]{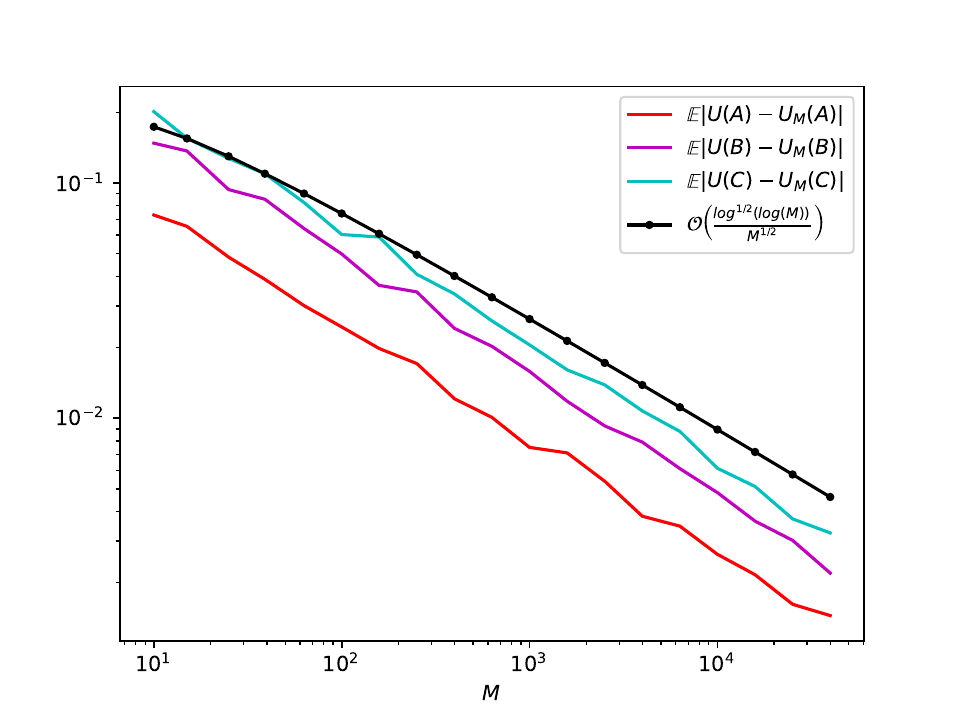}
    \label{U1_convergence}
\end{subfigure}
\caption{Expected utility in the one-dimensional linear problem: (a) true expected utility $U_{1}(\theta)$; (b) convergence behavior under empirical prior approximations.}
\label{fig:U1_combined}
\end{figure}

\subsection{Wasserstein-$2$ criterion and optimal transport}

For what follows, recall the Brenier's Theorem \ref{Optimal coupling}, which provides the theoretical backbone for computational approaches to evaluating the $U_2$ utility function in Euclidean spaces. When measures $\mu_1$ and $\mu_2$ have densities $\rho_1$ and $\rho_2$, respectively, with supports in an open set $\Omega\subset \R^n$, the potential function $\varphi$ from Theorem \ref{Optimal coupling} satisfies the Monge--Amp\'{e}re equation
\begin{equation}
\operatorname{det}\left(D^2 \varphi(x)\right)=\frac{\rho_{1}(x)}{\rho_{2}\left(\nabla \varphi(x)\right)}, \quad x \in \Omega.
\end{equation}
The Monge--Amp\`{e}re equations have been extensively studied in various settings (see e.g., \cite{cafarelli1992boundary,caffarelli1991some,caffarelli1992regularity,caffarelli1996boundary,gigli2011holder,philippis}). For optimal transport problems specifically, the choice of appropriate boundary conditions has been an active area of research \cite{pogorelov1994generalized}.
In particular, when $\Omega$ is a rectangle, the boundary conditions can be replaced by Neumann conditions and an average zero condition, leading to the boundary value problem \cite{froese2012numerical,froese2012_thesis}
\begin{equation}
\label{Monge-Ampere-BVP}
\begin{cases}
\operatorname{det}\left(D^2 \varphi \right)= \rho_{1}(x) / \rho_{2}(\nabla \varphi(x)), & x \in \Omega \\
\nabla \varphi(x) \cdot \mathbf{n}(x)=x\cdot \mathbf{n}(x), & x \in \partial \Omega \\
\varphi \text{ is convex}, & \\
\int_{\Omega} \varphi(x) dx=0 &
\end{cases}
\end{equation}
For more general domains, the transport boundary conditions are more challenging to implement numerically but can be replaced by Hamilton-Jacobi equations over the boundary (see \cite{benamou_britanny} for more details).

In what follows, we estimate the transport map between the prior and the posterior distribution through solving the system \eqref{Monge-Ampere-BVP} and use Brenier's theorem in concert with sampling schemes to estimate the expected utility.
Notice that due to symmetry of the Wasserstein distance, we can estimate transports $T^{y,\theta}_\sharp \mu_0 = \mu^y(\cdot; \theta)$ and $S^{y,\theta}_\sharp \mu^y(\cdot; \theta) = \mu_0$ by switching the role of $\rho_1$ and $\rho_2$. 
We now have
\begin{equation}
    \label{eq:U2_comp}
    U_2(\theta) = \E^{\pi(\cdot; \theta)} \E^\mu \norm{x-T^{y,\theta}(x)}^2
    = \E^{\nu(\cdot, \cdot;\theta)} \norm{x-S^{y,\theta}(x)}^2.
\end{equation}
Below, we provide two examples illustrating how the second identity with $S^{y,\theta}$ in \eqref{eq:U2_comp} can be utilized.
The expectations are approximated numerically via Monte-Carlo, or with Smolyak quadratures. In the case of the expectation respect to $\nu$, the term $\pi(x,y)$ can be replaced by the product $\pi(x)\pi(y\mid x)$ and both integrals can be estimated numerically.

\subsubsection{Example 1: Nonlinear forward mapping}

We consider the study case presented in \cite{huan2013simulation}, where the forward map of $y = {\mathcal G}(x;\theta) + \eta$ satisfies
\begin{equation}
{\mathcal G} : [0,1] \times [0,1]^2 \to \R^2; \quad (x, \theta) \mapsto \left[\begin{array}{l}
x^3 \theta_1^2+x \exp \left(-\left|0.2-\theta_1\right|\right) \\
x^3 \theta_2^2+x \exp \left(-\left|0.2-\theta_2\right|\right)
\end{array}\right].
\end{equation}
Moreover, the prior $\mu\sim\mathcal{U}([0, 1])$ is uniform and the additive noise has zero-mean Gaussian statistics with covariance $\Gamma = 10^{-4}{\rm Id}$. 

Notice that in the one-dimensional setup, the system \eqref{Monge-Ampere-BVP} simplifies and one has an explicit solution
\begin{equation*}
    S^{y,\theta}(x) = \frac 1{Z(y;\theta)}\int_0^1 \exp(-\Phi(x,y; \theta)) dx.
\end{equation*}
Here, the integration is performed with standard quadratures. 
The expectation over the joint distribution is $\nu(\cdot,\cdot; \theta)$ is then carried out using Smolyak quadratures. For the prior distribution and the likelihood, we used the Clenshaw-Curtis quadrature with 33 nodes and the Gauss-Hermite quadrature with 143 nodes, respectively. The expected utility was evaluated on a grid with $51\times 51$ design points. 

In Figure \ref{U_1D_figure}, the expected utility $U_2$ is plotted alongside the classical expected information gain criterion. The latter was generated using the same method as reported in \cite{duong2023stability}.
In EIG criterion the optimal design points are located in the points $A$ and $B$, followed by the corner $C$. Meanwhile, our W-OED criterion promotes the point $C$ in the corner as the optimal design. In order to compare the optimal design criteria, we generated synthetic data with ground truth $x = 0.8$ and with design nodes $A,B,C$. The corresponding posterior densities are showed on Figure \ref{U2_Samples} with the posterior being most concentrates at the point $C$.

\begin{figure}[htp]
\centering\includegraphics[scale = 0.7]{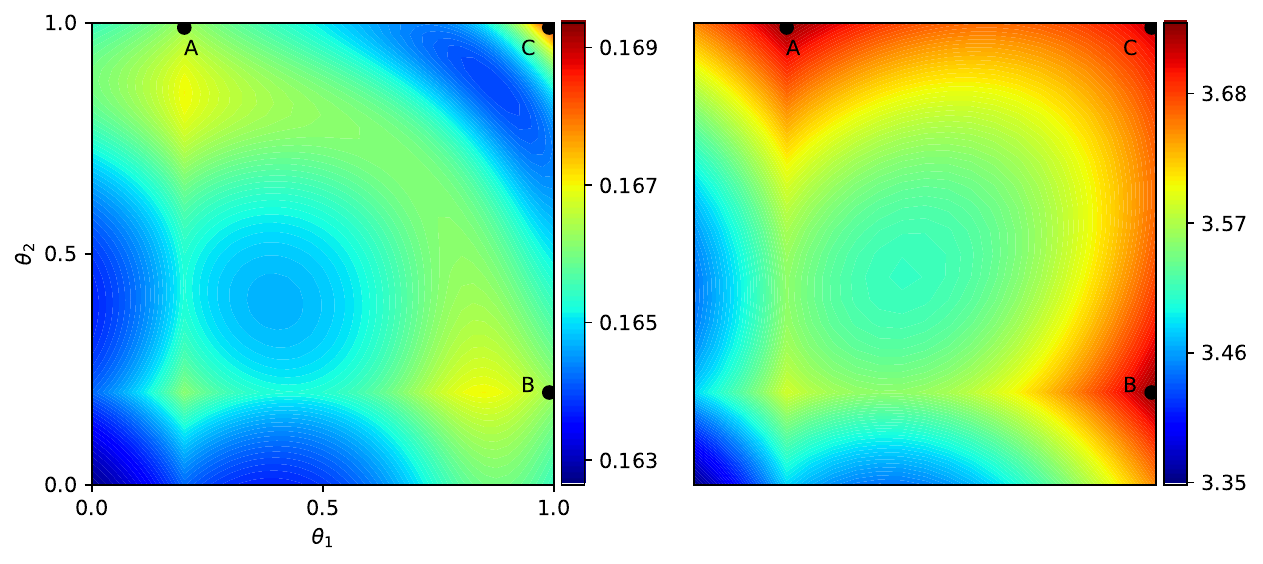}
\caption{Expected utility functions evaluated on the design domain for Example 1. Left: Wasserstein-$2$ information criterion. Right: Expected information gain criteria.}
\label{U_1D_figure}
\end{figure}

\begin{figure}[htp]
\centering\includegraphics[scale = 0.7]{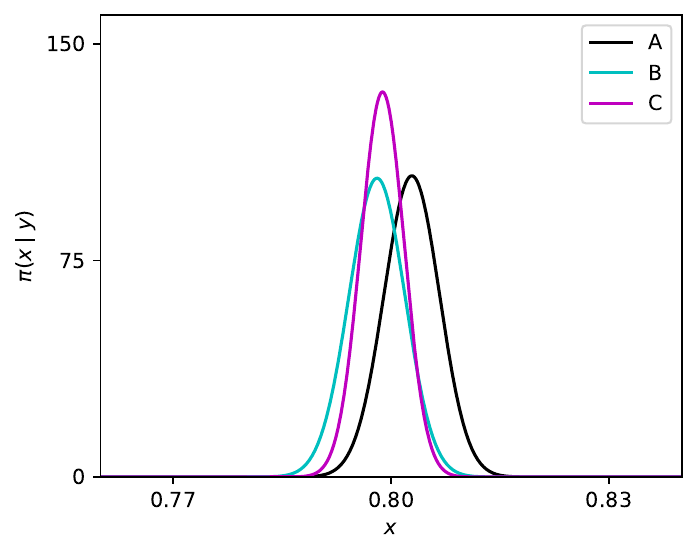}
\caption{Posterior densities for Example 1 based on synthetic data with ground truth $x = 0.8$ and with design nodes $A,B,C$.}
\label{U2_Samples}
\end{figure}

\subsubsection{Example 2: Heat diffusion}

Let us consider heat diffusion with a source $S$ and Neumann boundary conditions
\begin{equation}
\begin{aligned}
\frac{\partial v}{\partial t} & =\Delta v+S(\cdot, x), & & (z, t) \in \Omega \times [0,0.4] \\
\nabla v \cdot n & =0, & & (z, t) \in \partial \Omega \times [0,0.4] \\
v(z, 0) & =0, & & z \in \Omega,
\end{aligned}
\end{equation}
where $\Omega$ is the unit square.
Here, $n$ denotes the normal vector to the boundary and the source term satisfies
\begin{equation}
S(z, t, x)= \begin{cases}\frac{1}{\pi h^2} \exp \left(-\frac{\|z-x\|^2}{2 h^2}\right), & 0 \leq t<\tau \\ 0, & t \geq \tau\end{cases}
\end{equation}
with parameters $h = 0.05$ and $\tau = 0.3$.
We note that the same example was considered in \cite{huan2014gradient}.

Here, we consider the design task of identifying optimal sensor location $\theta\in\Omega$ for recovering the source location $x$
by observations performed at times $t \in \{ 0.08, 0.16, 0.24, 0.32, 0.40\}$. In other words, the inverse problem is to recover $x$ based on the data vector $y = [ v(\theta,t_1) ... v(\theta,t_5)]\in \R^{5}$. Consequently, the forward mapping ${\mathcal G} : \Omega \times \Omega \to \R^5$ with $y = {\mathcal G}(x; \theta)$. In our numerical implementation, we reduce the computational cost of evaluating ${\mathcal G}$ by approximating it with a surrogate model based on polynomial chaos expansions with degree 8 following \cite{huan2014gradient}. 

In our implementation, we construct the mapping $S^{y,\theta}$ in \eqref{eq:U2_comp} and evaluate the expectation over the joint distribution $\nu$ using Smolyak quadratures. In particular, we apply a Clenshaw--Curtis configuration of 34 nodes and Gauss-Hermite with 117 nodes for $\mu$ and $\pi(y\mid x)$ respectively. Moreover, the expected utility was constructed on a grid of $23\times 23$ nodes on the design domain $\Omega$. 

The mapping $S^{y,\theta}$ or more precisely, an approximation of its potential is obtained by solving the system \eqref{Monge-Ampere-BVP} with Radial Basis Function (RBF) approximations. RBF based methods for the Monge--Ampere problem have been an active topic of study 
(see e.g. \cite{jianyu2003numerical, liu2013solving, liu2014iterative,liu2020multiscale,liu2023meshfree}).
In particular, we focus on to a finite difference approach with convexity restrictions (see, e.g., \cite{benamou2012viscosity,engquist2014application,froese2012numerical,froese2012_thesis}).

Following Kansa's asymmetric collocation approach \cite{kansa1990multiquadrics_1,kansa1990multiquadrics_2}, we select two collocation point sets $\lbrace x^{k} \rbrace_{k=1}^{M_I}\subset \Omega$ and $\lbrace x^{k} \rbrace_{k=M_{I}+1}^{M_r}\subset \partial\Omega$, and an ansatz of the form:
\begin{equation}
\label{ansatz}
    \hat \phi^{y,\theta}(x) = \hat \phi(x) = \sum_{k=1}^{M_{r}}\lambda_{k}\psi_{k}(x) + \sum_{j=1}^{M_p}\alpha_{j}p_{j}(x),
\end{equation}
where $\psi_k(x) := \norm{x-x^k}^{4}\log(\norm{x-x^k})$ are the second order thin plate splines, and $\lbrace p_j \rbrace_{j=1}^{M_p}$ are a second-order polynomial basis $\lbrace 1, x_{1}, x_{2}, x_{1}^{2}, x_{1}x_{2},x_{2}^{2}\rbrace$ on $\Omega$.
Moreover, in the two-dimensional case, we can ensure convexity of $\hat \phi$ by imposing the constraint:
\begin{equation}
\label{Kansa_3}
    \nabla^{2}\hat \phi(x) > 0, \quad x\in\partial\Omega.
\end{equation}
Here, we used $M_I = 144$ and $M_r = 192$ the number of nodes in the Monge-Ampere solution with RBFs. 

The resulting utility function $U_2$ was plotted on the figure \ref{U_2D_figure}, next to the EIG criterion. The last one was simulated using the methods from \cite{duong2023stability}. Both cases are symmetric about the midpoint $(0.5, 0.5)$ in the diagonal, horizontal, and vertical directions, as anticipated from the problem setup. Both expected utilities increase when approaching boundaries. However, the Wasserstein criterion prefers the boundary mid-points $\theta \in \{(0.5,0), (1,0.5), (0.5,1), (0,0.5)\}$, while the EIG criterion prefers the corner points.

In order to compare both design criteria, we generated synthetic data with ground truth values $x\in\lbrace (0.09,0.22), (0.25,0.75), (0.75,0.25), (0.75,0.75) \rbrace$ comparing design values $\theta = (0,0)$ (preferred by EIG) and $\theta = (0.5,0)$ (preferred by WOED). The corresponding posterior are shown in Figure \ref{fig:posterior_all}.

While the examples cannot provide a complete picture of the different preferences by the two design criteria, it is still evident that EIG prefers high precision reconstruction from a subset of points (close to the corner at origin) with the cost of reduced precision from other corners (close to $(0, 1)$ and $(1,0)$), while WOED prefers an averaged performance between the neighbouring corners ($0,0$ and $(1,0)$).

\begin{figure}[htp]
\centering\includegraphics[scale = 0.7]{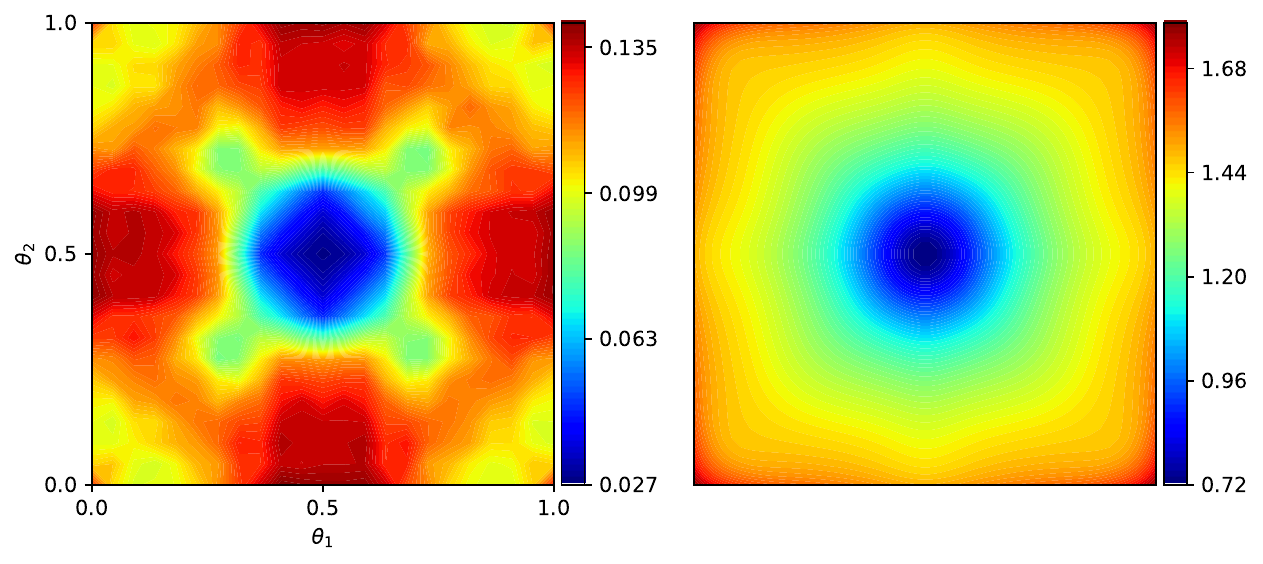}
\caption{Expected utility functions $U(\theta)$ for Example 2. Left: Wasserstein-2 distance criteria. Right: D-OED criteria.}
\label{U_2D_figure}
\end{figure}

\begin{figure}[htp]
\centering
\includegraphics[scale=0.5]{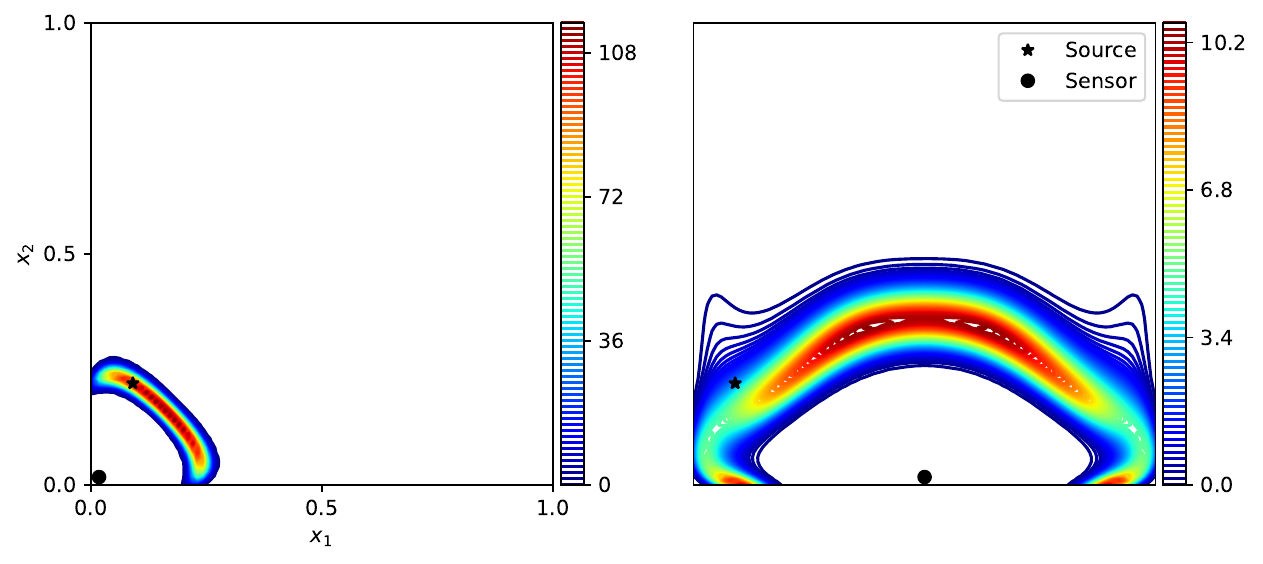}
\vspace{0.5em}
\includegraphics[scale=0.5]{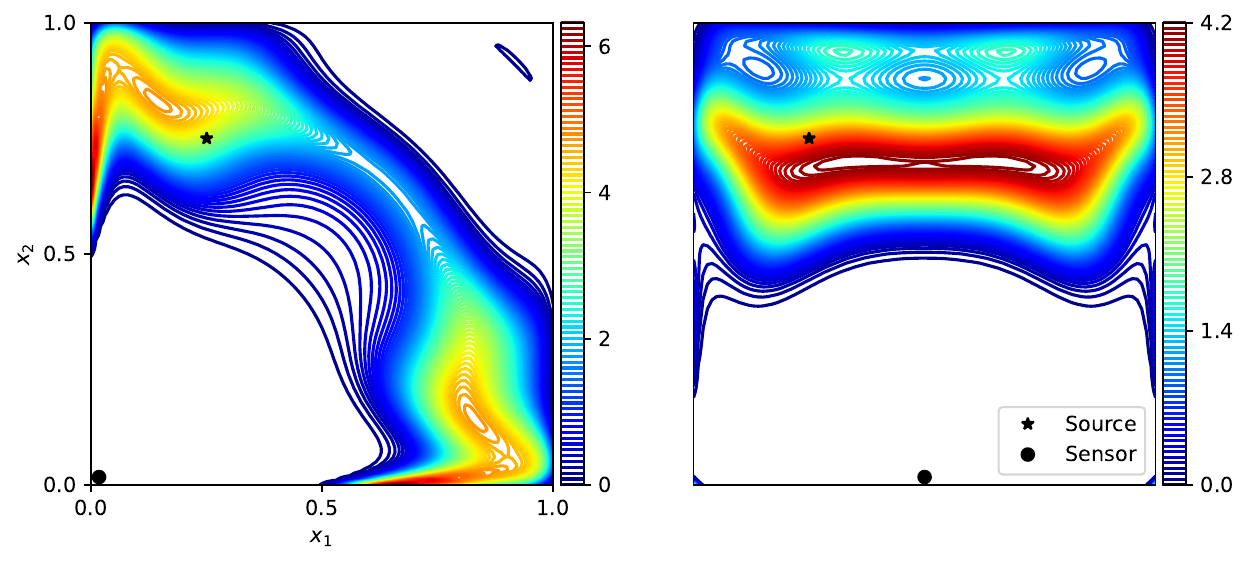}
\vspace{0.5em}
\includegraphics[scale=0.5]{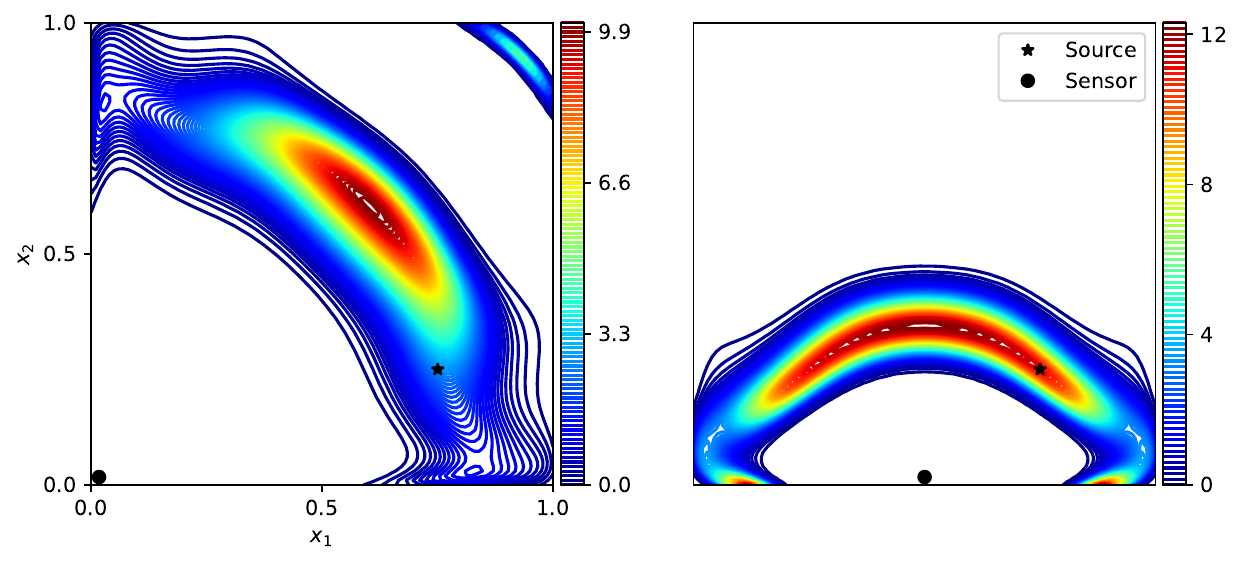}
\vspace{0.5em}
\includegraphics[scale=0.5]{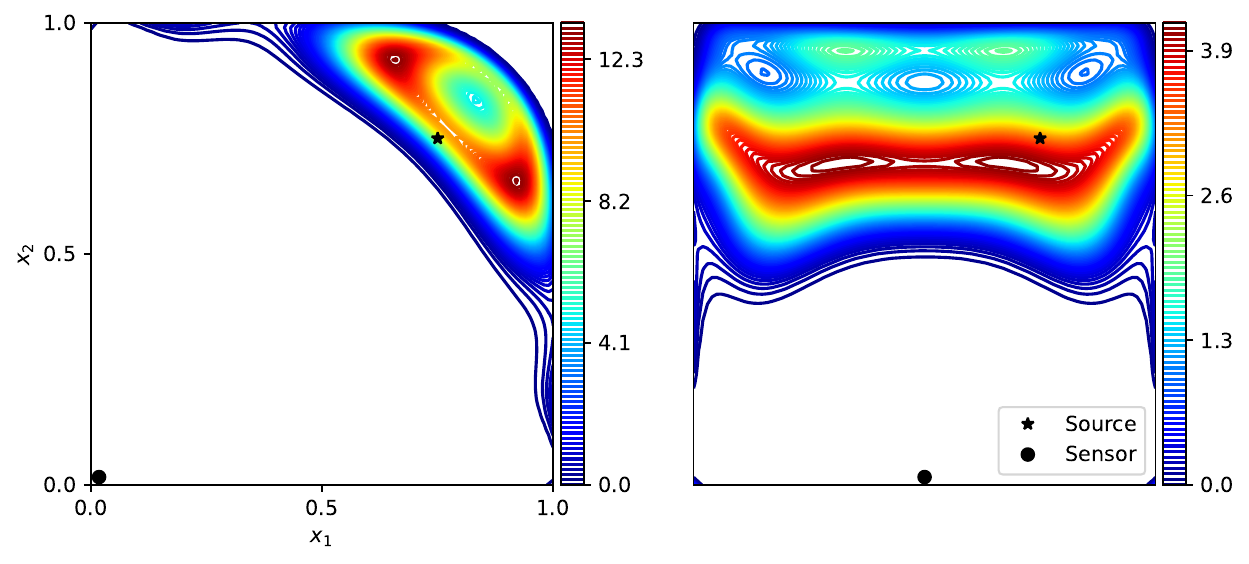}
\caption{Posterior densities generated at optimal design nodes. Each row shows posteriors for different values of $x$: (a) $x = (0.09, 0.22)$, (b) $x = (0.25, 0.75)$, (c) $x = (0.75, 0.25)$, and (d) $x = (0.75, 0.75)$. In each subfigure, the left panel corresponds to the D-OED optimal design and the right to the Wasserstein-2 distance optimal design.}
\label{fig:posterior_all}
\end{figure}

\section{Conclusions}
\label{sec:conclusions}

In this paper, we have proposed a novel Wasserstein distance based information criterion in the context of Bayesian optimal experimental design. We highlighted several key properties that make this utility particularly appealing for large-scale inference problems, including its scalability to infinite-dimensional Hilbert space settings and its validity as an information criterion in the sense of Ginebra \cite{ginebra2007}. In the case of linear-Gaussian models, we showed that the Wasserstein-$2$ criterion admits a closed-form expression that is closely related to, but distinct from, Bayesian A-optimality and its weighted variants. This expression nonetheless enables efficient computation in the linear-Gaussian setting.
%
Our stability analysis further supports the robustness of the proposed approach for nonlinear inverse problems with a Gaussian likelihood model. Finally, we demonstrated the practical computability of the Wasserstein criterion using standard algorithms from the optimal transport literature.

The paper opens several promising directions for future research. First, it would be useful to understand the stability of the method for more general likelihood models. Furthermore, the coupling of the Lipschitz constant $L_1$ and the exponential moment of the prior with constant $L_2$ (Section~\ref{sec:stability}) is restrictive, and more work is needed to relax this assumption. Second, our computational methods so far address only low-dimensional examples. Future work should develop efficient methods for high-dimensional problems. 
To this end, we point to recent papers on learning \textit{conditional optimal transport maps} from samples, which represent such maps as gradients of partially input-convex neural networks \cite{huang2020convex,bunne2022supervised,baptista2024conditional,wang2025efficient}. Conditional OT maps are directly suited to the Wasserstein OED setting: rather than solving the Monge--Amp\`ere equation many times, to compute the Wasserstein utility for many distinct values of the data $y$, these maps instead encode $y$-parameterized OT maps from prior to posterior. As explained in \cite[Section 2.3]{baptista2024conditional}, the weighted $L^2$ norm of such a map corresponds precisely to the expected Wasserstein utility proposed here, at least for $p=2$, with the possibility of extensions to other values of $p$.

\bibliographystyle{plain} 
\bibliography{bibliography}       %

\end{document}